\documentclass[12pt, reqno]{amsart}
\usepackage{graphicx} 
\usepackage{color}
\usepackage{ulem} 
\usepackage{hyperref}
\newtheorem{theorem}{Theorem}[section]
\newtheorem{lemma}{Lemma}[section]
\newtheorem{remark}{Remark}[section]
\newtheorem{example}{Example}[section]
\newtheorem{defn}{Definition}[section]
\newtheorem{proposition}{Proposition}[section]

\numberwithin{equation}{section}
\numberwithin{figure}{section}

\begin{document}

\newcommand{\E}{\mathbb{E}}
\newcommand{\Eof}[1]{\mathbb{E}\left[ #1 \right]}
\newcommand{\tE}{\tilde{\mathbb{E}}}
\newcommand{\tEof}[1]{\tilde{\mathbb{E}}\left[ #1 \right]}
\renewcommand{\P}{\mathbb{P}}
\newcommand{\PP}{\mathbb{P}}
\newcommand{\tP}{\tilde{\mathbb{P}}}
\newcommand{\tr}{{\rm tr}}
\newcommand{\tW}{\tilde{W}}
\newcommand{\cF}{\mathcal{F}}
\newcommand{\R}{\mathbb{R}}
\newcommand{\sigl}{\sigma_\ell}
\newcommand{\BS}{\rm BS}
\newcommand{\p}{\partial}
\newcommand{\var}{{\rm var}}
\newcommand{\cov}{{\rm cov}}
\newcommand{\beaa}{\begin{eqnarray}}
\newcommand{\eeaa}{\end{eqnarray}}
\newcommand{\bea}{\begin{eqnarray}}
\newcommand{\eea}{\end{eqnarray}}
\newcommand{\ben}{\begin{enumerate}}
\newcommand{\een}{\end{enumerate}}
\newcommand{\beq}{\begin{equation}}
\newcommand{\eeq}{\end{equation}}

\newcommand{\bs}{\mathbf{s}}
\newcommand{\e}{\varepsilon}
\newcommand{\Q}{\mathbb{Q}}
\newcommand{\sech}{\mathrm{sech}}

%
%
%

\begin{titlepage}

\begin{center}
 \large \bf Most-likely-path in Asian option pricing under local volatility models
\end{center}
\vspace{0.5cm}

\begin{center}
  Louis-Pierre Arguin \\
  Department of Mathematics, Baruch College, CUNY \\
  1 Bernard Baruch Way, New York, NY10010 \\
  e-mail: \textsf{louis-pierre.arguin@baruch.cuny.edu}
\end{center}

\vskip5mm

\begin{center}
  Nien-Lin Lu \\
  Department of Mathematical Sciences, Ritsumeikan University \\
  Noji-higashi 1-1-1, Kusatsu, Shiga 525-8577, Japan \\
  e-mail: \textsf{nienlin1126@gmail.com}
\end{center}

\vskip5mm

\begin{center}
  Tai-Ho Wang \\
  Department of Mathematics, Baruch College, CUNY \\
  1 Bernard Baruch Way, New York, NY10010 \\
  e-mail: \textsf{tai-ho.wang@baruch.cuny.edu}
\end{center}

\vskip1cm
\begin{center}
 {\bf Abstract}
 \end{center}
\vskip 0.7cm

This article addresses the problem of approximating the price of options on discrete and continuous arithmetic averages of the underlying, i.e., discretely and continuously monitored Asian options, in local volatility models.
A ``path-integral"-type expression for option prices is obtained using a Brownian bridge representation for the transition density between consecutive sampling times and a Laplace asymptotic formula. In the limit where the sampling time window approaches zero,
the option price is found to be approximated by a constrained variational problem on paths in time-price space.
We refer to the optimizing path as the {\it most-likely path} (MLP).
An approximation for the implied normal volatility follows accordingly.
The small-time asymptotics and the existence of the MLP are also  rigorously recovered using large deviation theory.

\vskip 5mm


%

\vskip 1cm



\end{titlepage}



\title[Most-likely-path in Asian option pricing]{Most-likely-path in Asian option pricing under local volatility models}

\author[L.-P. Arguin, N.-L. Liu, and T.-H. Wang]{Louis-Pierre Arguin, Nien-Lin Liu, and Tai-Ho Wang}

\address{Louis-Pierre Arguin \newline
Department of Mathematics, Baruch College, CUNY \newline
1 Bernard Baruch Way, New York, NY10010
}
\email{louis-pierre.arguin@baruch.cuny.edu}

\address{Nien-Lin Liu \newline
Department of Mathematical Sciences, Ritsumeikan University \newline
Noji-higashi 1-1-1, Kusatsu, Shiga 525-8577, Japan
}
\email{nienlin1126@gmail.com}

\address{Tai-Ho Wang \newline
Department of Mathematics \newline
Baruch College, The City University of New York \newline
1 Bernard Baruch Way, New York, NY10010 
}
\email{tai-ho.wang@baruch.cuny.edu}

\maketitle


\begin{abstract}
This article addresses the problem of approximating the price of options on discrete and continuous arithmetic averages of the underlying, i.e., discretely and continuously monitored Asian options, in local volatility models.
A ``path-integral"-type expression for option prices is obtained using a Brownian bridge representation for the transition density between consecutive sampling times and a Laplace asymptotic formula. In the limit where the sampling time window approaches zero,
the option price is found to be approximated by a constrained variational problem on paths in time-price space.
We refer to the optimizing path as the {\it most-likely path} (MLP).
An approximation for the implied normal volatility follows accordingly.
The small-time asymptotics and the existence of the MLP are also  rigorously recovered using large deviation theory.
\end{abstract}


\renewcommand{\O}{\mathcal{O}}
\newcommand{\dt}{\Delta t}
\newcommand{\mt}{\mathbf{t}}
\allowdisplaybreaks

%
%

\section{Introduction}

Asian options, also known as average price options, are among the most liquidly traded exotic options in commodities such as agriculture, energy, and fixed incomes markets. Nowadays average price options represent a high percentage of options on oil; some are directly on the futures contracts of oil, while others on spreads between two types of oil futures.
Asian options are also commonly used as a risk management vehicle for, owing to its averaging feature, a) the underlying average price is more difficult to manipulate; b) the average price is less sensitive to abrupt shocks; and c) such options are cheaper than similar vanilla options.

There is a rich literature on the problem of pricing Asian options, in part due to the difficulty of finding analytical solutions even for simpler cases like the Black-Scholes model.
We mention only the following few and refer the interested readers to the references therein.
To our knowledge, \cite{kemna-vorst} was the first published work tackling the problem of Asian option pricing in Black-Scholes model. As numerically pricing Asian options by Monte Carlo simulations is concerned, \cite{kemna-vorst} also introduced a variance reduction technique by using the price of geometric average option, whose analytic form is readily obtained, as a control variate. See also \cite{curran} for further analysis and extension to portfolio options on the technique of conditioning. Further development and improvement of Monte Carlo schemes since then were followed up by \cite{broadie-glasserman} (for a direct method of estimating the Greeks of an Asian option by Monte Carlo simulation, see Section 4.2 on P.275), \cite{vd} (combining control variate and change of measure/likelihood ratio), and \cite{ghs} (importance sampling), etc.
Attempts to find closed or semi-closed form expressions for Asian option pricing in Black-Scholes models first appeared in the seminal work of \cite{gy}. Among other interesting results in \cite{gy}, the Laplace transform of the Asian option price with respect to time to maturity is derived and has been known as the celebrated Geman-Yor formula. An extension of the Geman-Yor approach to a jump diffusion model can be found in \cite{cai-kou}. However, numerical inversion of the Geman-Yor formula was shown slow and needed to be handled with care, see for example the discussions in \cite{dufresne} and \cite{fmw}.
Analytical approximation of the risk neutral density for the average price, which in turn yields an approximation of the Asian option, under Black-Scholes model dates back to the work of \cite{turnbull-wakeman} and \cite{ritchken}. Both papers applied the Edgeworth expansion to the density of discretely monitored average price around lognormal distribution and obtained Black-Scholes type formula (up to correction terms) for the price of Asian option. Such analytical approximations are more appealing in practice than the Monte Carlo solutions because explicit expressions for the Greeks are readily accessible.
Finally, PDE method is pursued by \cite{rogers} and \cite{vercer1} (see also \cite{vercer2} for a more recent development) after an ingeniously chosen change of variable.

Literatures on the pricing of Asian options under more general dynamics for the underlying security such as local volatility or stochastic volatility models, as opposed to those in the Black-Scholes model, are comparatively little. Less ambitious approaches to the pricing of Asian options include arbitrage free bounds as in \cite{albrecher} and \cite{rogers}, and also approximative and asymptotic solutions such as \cite{dassios}, \cite{fpp}, and \cite{ppp}. Approximations resorting to asymptotic expansions are mostly It\^o-Taylor type expansion based as originated from the work of \cite{kunitomo-takahashi}, see also \cite{cai-li-shi} for a more recent development along this line. Despite being straightforward but tedious, such expansions usually require calculations up to third or fourth order in order to achieve satisfactory accuracy. Finally, though not directly related to the current paper, Asian option pricing under stochastic volatility models is discussed in \cite{fouque} and \cite{takahashi}.

In the current paper, we address the problem of approximating the price of options on the discrete arithmetic average, and its continuous-time limit, with the underlying following a local volatility model. For the discretely monitored Asian option, we assume that the average is over a set of equally spaced discrete time samples before expiry. The application of the Brownian bridge representation for the transition density (see Theorem \ref{thm:BBrep-diffusion} below) obtained in \cite{hk-prob} between consecutive sampling time points leads to a ``path-integral" type expression for the Asian option price, see \eqref{eqn:Asian-call}.
A direct application of a Laplace asymptotic formula (in this case high-dimensional, see Lemma \ref{lma:laplace}) yields an approximation of the option price (see Theorem \ref{thm:approx-discrete-Asian}).
In the limit where the sampling time window approaches zero, the leading order term (in small time to expiry) can be expressed as a constrained variational problem of finding an optimal path, referred to as the {\it most-likely-path} (MLP), in the time-price space. 
An approximation of the continuously monitored Asian option price is obtained once the variational problem is solved, see Definition \ref{def:mlp} and Theorem \ref{thm:cts-Asian}.
The MLP approximation coincides with a rigorous derivation of the leading order based on a recent extension of the Freidlin-Wentzell theorem for a large class of models, see Theorem \ref{thm: ld}.
As for implied volatility, we opt to use the Bachelier model as benchmark rather than the Black-Scholes model because of the lack of analytical expression for Asian options in the Black-Scholes model. Such defined implied volatility in the European option case is sometimes referred to as the {\it implied normal volatility} in practice. By comparing corresponding expansions from the benchmark Bachelier model and from the local volatility model, we obtain, as the main result of the paper, the lowest order approximation of the implied normal volatility for Asian option in Theorem \ref{thm:imp-vol-approx}.

The paper is organized as follows. Section 2 lays out the model and provides derivations of the Laplace type approximation for discretely monitored Asian calls and the most-likely-path approximations for continuously monitored Asian call options. Section \ref{sec:LDP} is devoted to a rigorous derivation of the asymptotic behavior obtained in Section \ref{sec:Asian-option-asym} based on large deviation theory.
It was drawn to our attention that a rigorous proof of the most-likely-path approximation for the price of continuously monitored Asian option by the theory of large deviation has been done independently in \cite{pirjol-zhu}. Finally, we conclude by presenting numerical tests of the most-likely-path approximation.

Throughout the text, $(W_t,t\geq 0)$ denotes the standard Brownian motion defined on the filtered probability space $(\Omega,(\cF_t)_{t\geq 0},\P)$ satisfying the usual conditions.
Dot will always refer to the partial derivative with respect to the time variable and prime to the space variable $s$.

%
%

\section{Asian option pricing in local volatility models} \label{sec:Asian-option-asym}

In this section, we derive asymptotic expansions of the prices of both discretely and continuously monitored Asian calls in local volatility models.
First, the most-likely-path approximation in the discrete case is obtained in Theorem \ref{thm:approx-discrete-Asian} using a Brownian bridge representation for the transition density obtained in \cite{hk-prob} (see Theorem \ref{thm:BBrep-diffusion} below) and  a high-dimensional Laplace asymptotic formula (Lemma \ref{lma:laplace}). Second, an expression for the most-likely-path approximation in the continuous case is derived in Theorem \ref{thm:cts-Asian} by taking a formal limit of the approximation in the discrete case.
This provides a heuristic derivation for the leading order of the most-likely-path approximation for continuously monitored Asian calls.

We assume that the underlying asset $S$ follows the local volatility model
\begin{equation}
\label{eqn:local-vol}
  d S_t = S_t \sigl(S_t,t) dW_t = a(S_t,t) dW_t\qquad S_0=s_0 > 0\ .
\end{equation}
We will suppose that the diffusion function $a(s,t)$ is strictly positive (except possibly at $s=0$ where it could be $0$), and grows {\it at most linearly} in $s$: there exists $C>0$ such that
\begin{equation}
\label{eqn: bounded}
0\leq a(s,t)\leq C(1+|s|) \text{ for all $t\in [0,T]$ and for all $s\in \R$; }
\end{equation}
and is {\it locally Lipschitz}: for every $R>0$, there exists $C_R$ such that for all $t\in[0,T]$ and for every $x,y\in \R$ with $|s|,|s'|<R$
\begin{equation}
\label{eqn: lipschitz}
 |a(s,t)-a(s',t)|\leq C_R|s-s'|\ .
\end{equation}
These assumptions are sufficient to imply the existence and uniqueness of a strong solution to the SDE \eqref{eqn:local-vol}.
Such solutions will also satisfy a {\it large deviation principle} based on the work in \cite{chiarini-fischer} as explained in Section \ref{sec:LDP}.

Let $p(T,s_T|t,s_t)$, $t<T$, be the transition probability density from $(t,s_t)$ to $(T,s_T)$ for the local volatility model \eqref{eqn:local-vol}.
Consider the Lamperti transformation from $s$ to $x$
\begin{equation}
\label{eqn:lamperti}
x = \varphi(s,t) = \int_{s_0}^s \frac{d\xi}{a(\xi,t)} \quad \mbox{ for } s > 0 \ .
\end{equation}
Since $a$ is assumed locally Lipschitz and is strictly positive except at $s=0$, the transformation is well-defined for all $s > 0$, except possibly at $s = 0$.
The following representation for the transition density $p$ is derived in \cite{hk-prob}.
\begin{theorem} \label{thm:BBrep-diffusion}
Let $S=(S_t, t\geq 0)$ be the diffusion process given by \eqref{eqn:local-vol}.
Define the function $h$ by $h(x,t) = \varphi_t(s,t) - a_s(s,t)/2$, with $s = \varphi^{-1}(x,t)$, where $\varphi$ is the Lamperti transformation \eqref{eqn:lamperti}
and subindices refer to corresponding partial derivatives.
Then the transition density $p$ of $S$ from $(t,s_t)$ to $(T,s_T)$ has the representation:
\begin{equation} \label{eqn:BB-diffusion-NoIto}
  p(T,s_T|t,s_t) = \frac{g(T-t,\varphi(s_T,T)-\varphi(s_t,t))}{a(s_T,T)} \tE_{\varphi(s_t,t),\varphi(s_T,T)}\left[ e^{\int_t^T h(X_s,s)dX_s - \frac12\int_t^T h^2(X_s,s)ds} \right]
\end{equation}
where $g$ denotes the centered Gaussian density with variance $t$: $g(t,\xi) = \exp(-\xi^2/2t)/\sqrt{2\pi t}$ and
 $\tE_{x,y}[\cdot]$ is the expectation under the Brownian bridge measure from $x$ to $y$.

Equivalently, if $H$ is an antiderivative of $h$ with respect to $x$, namely, $\p_x H(x,t) = h(x,t)$, for all $x$ and $t$, then
\bea
 \label{eqn:BB-diffusion-Ito}
  p(T,s_T|t,s_t) &=& \frac{g(T-t,\varphi(s_T,T)-\varphi(s_t,t))}{a(s_T,T)} e^{H(\varphi(s_T,T),T) - H(\varphi(s_t,t),t)} \times \nonumber \\
  && \quad \tE_{\varphi(s_t,t),\varphi(s_T,T)}\left[ e^{- \frac12\int_t^T h^2(X_s,s) + h_x(X_s,s) + 2H_t(X_s,s) ds} \right] .
\eea
\end{theorem}
For notational simplicity, hereafter we shall denote the expectation term in \eqref{eqn:BB-diffusion-NoIto} as
\begin{equation}
\Psi(s_t,s_T) = \tE_{\varphi(s_t,t),\varphi(s_T,T)}\left[ e^{\int_t^T h(X_s,s)dX_s - \frac12\int_t^T h^2(X_s,s)ds} \right].
\end{equation}
With this notation, we also have for the term in \eqref{eqn:BB-diffusion-Ito}
\begin{equation}
\begin{aligned}
\Psi(s_t,s_T) &=e^{H(\varphi(s_T,T),T) - H(\varphi(s_t,t),t)} \times \\
&\hspace{1cm}\tE_{\varphi(s_t,t),\varphi(s_T,T)}\left[ e^{- \frac12\int_t^T h^2(X_s,s) + h_x(X_s,s) + 2H_t(X_s,s) ds} \right].
\end{aligned}
\end{equation}

\subsection{Small time asymptotic for discretely monitored Asian call}
Assume an (arithmetic) Asian call is sampled discretely at the time points $t_1 < t_2 < \cdots < t_n$ with $t_0 = 0$ and $t_n = T$. In other words, the payoff of such an Asian call is
\begin{equation}
  \left( \frac1{n}\sum_{i=1}^{n} S_{t_i} - K \right)^+,
\end{equation}
where the time interval between the sampling points is assumed equal, i.e., $t_{i} - t_{i-1} = \Delta t = T/n$ for $i=1,\cdots,n$.
By using the Brownian bridge representation \eqref{eqn:BB-diffusion-NoIto} or \eqref{eqn:BB-diffusion-Ito}, the joint density for $S_{t_1}, \cdots, S_{t_n}$ can be written as
\beaa
  && p(t_1, s_{t_1}|t_0, s_{t_0}) p(t_2, s_{t_2}|t_1, s_{t_1})  \cdots p(t_n , s_{t_n}|t_{n-1}, s_{t_{n-1}}) \\
  &=& \prod_{i=1}^n g(\Delta t,\varphi(s_{t_i},t_i)-\varphi(s_{t_{i-1}},t_{i-1})) \, \frac{\Psi(s_{t_i},s_{t_{i-1}})}{a(s_{t_i},t_i)}.
\eeaa
In this notation, the price $C=C(s_0,0;K,T)$ of a discretely monitored Asian call struck at $K$ can be written as
\begin{equation}
 \label{eqn:Asian-call}
\begin{aligned}
C &= \E\left[ \left( \frac1n\sum_{i=1}^n S_{t_i} - K \right)^+ \right] \\
&= \iint \left(\frac1n\sum_{i=1}^n s_{t_i} - K \right)^+ \prod_{i=1}^n p(t_i,s_{t_i}|t_{i-1},s_{t_{i-1}}) ds_{t_1}\cdots ds_{t_n} \\
&= \frac1{\left( 2\pi\Delta t\right)^{\frac n2}}\iint \left(\frac1n\sum_{i=1}^n s_{t_i} - K \right)^+ e^{-\frac{D(
\bs,\mathbf t)}{\Delta t}} W(\bs,\mathbf t) d\bs,
\end{aligned}
\end{equation}
where
\begin{equation}
\begin{aligned}
D(\bs,\mathbf t) &= \frac12\sum_{i=1}^n \left|\varphi(s_{t_i},t_i) - \varphi(s_{t_{i-1}},t_{i-1})\right|^2, \quad
  W(\bs,\mathbf t) = \prod_{i=1}^n \frac{\Psi(s_{t_{i-1}},s_{t_i})}{a(s_{t_i},t_i)},
\end{aligned}
\end{equation}
for $\mathbf t = (t_1,\cdots,t_n)$, $\bs = (s_{t_1}, \cdots, s_{t_n})$ and $d\bs = ds_{t_1}\cdots ds_{t_n}$.

Hence, the evaluation of the Asian call price $C$ becomes the computation of the multidimensional integral \eqref{eqn:Asian-call}.
The following Laplace asymptotic formula will be applied to estimate the $n$-dimensional integral in \eqref{eqn:Asian-call} for small $\Delta t$.

\begin{lemma}(Laplace asymptotic formula) \label{lma:laplace} \\
Let $R$ be a closed set in $\R^n$ with nonempty and smooth boundary $\p R$. Suppose $\theta$ is a continuous function in $R$ and attains its minimum uniquely at $x^* \in \p R$ and, given any $\epsilon > 0$, there exists $\delta > 0$ such that $\theta(x) \geq \theta(x^*) + \delta$ for all $x \in R \setminus B_\epsilon(x^*)$, where $B_\epsilon(x^*) = \{x:|x - x^*| < \epsilon \}$ is the open ball of radius $\epsilon$ centered at $x^*$. Assume that $f$ is integrable in $R$, i.e., $\int_R |f(x)|dx < \infty$ and that $f$ vanishes identically in $R^c$ and on the boundary $\p R$ but the inward normal directional derivative of $f$ at $x^*$ is nonzero.
Then we have the asymptotic expansion as $\tau \to 0^+$
\begin{eqnarray} 
&& \int_R e^{-\frac{\theta(x)}{\tau}} f(x) dx \label{lap-asymp} \\
&=& \frac{(2\pi)^{\frac{n-1}2} \tau^{\frac{n + 3}2} e^{-\frac{\theta(x^*)}{\tau}}}{\sqrt{\det \p_{\mathbf t}^2\theta(x^*)}|\nabla\theta(x^*)|}
\left[ \frac{\nabla f(x^*)\cdot\nabla\theta(x^*)}{|\nabla\theta(x^*)|^2} + \frac12 \tr\left\{\p_{\mathbf t}^2 f(x^*) \left[\p_{\mathbf t}^2\theta(x^*)\right]^{-1}\right\} + \O(\tau) \right], \nonumber
\end{eqnarray}
where $\p_{\mathbf t}^2 f(x^*)$ and $\p_{\mathbf t}^2 \theta(x^*)$ are the Hessian matrices of $f$ and $\theta$ respectively in the tangential direction to $R$ at $x^*$.
\end{lemma}
The proof of the lemma is standard and straightforward. See for instance Section 8.3 in \cite{BH}. 

Now we apply the Laplace asymptotic formula \eqref{lap-asymp} to the multidimensional integral \eqref{eqn:Asian-call} by taking $\theta = D(\bs,\mt)$, $f = \left(\frac1n \sum_{i=1}^n s_{t_i} - K \right) W(\bs,\mt)$, and $R$ as the half space $R := \left\{\bs: 1/n\sum_{i=1}^n s_{t_i} \geq K \right\}$. The crucial step in applying the Laplace asymptotic formula \eqref{lap-asymp} is the determination of the minimum point of $D$ in the half space $R$, which boils down to solving the constrained optimization problem:
\bea
  && \min_{\bs} \frac12 \sum_{i=1}^n \left|\varphi(s_{t_i},t_i) - \varphi(s_{t_{i-1}},t_{i-1})\right|^2 \label{eqn:discrete-obj}
\eea
subject to
\begin{equation}
  \frac1n \sum_{i=1}^n s_{t_i} \geq K. \label{eqn:discrete-constr}
\end{equation}
\begin{remark}
{\rm We shall assume $s_0 < K$ in the following for if $s_0 \geq K$, the value of the constrained optimization problem \eqref{eqn:discrete-obj}:\eqref{eqn:discrete-constr} is 0 since one can simply take $s_{t_1} = s_{t_2} = \cdots = s_{t_n} = s_0$. Then $s_1 + \cdots + s_n = n s_0 \geq n K$ and $\varphi(s_{t_i},t_i) = \varphi(s_{t_0},t_i) = 0$ for all $1 \leq i \leq n$. Thus the objective function in \eqref{eqn:discrete-obj} attains its global minimum $0$.
 }
\end{remark}
\begin{remark}
{\rm We show in Section \ref{sec:appendix-convexity} Appendix I that, for $\Delta t$ small enough, the objective function in \eqref{eqn:discrete-obj} is in fact convex, which in turn implies that the minimizer, if there exists any, is unique since the constraint is a linear inequality.
 }
\end{remark}

We summarize the result for the price of a discretely monitored Asian call in Theorem \ref{thm:approx-discrete-Asian} whose proof in the time homogeneous case is simply a direct application of the Laplace asymptotic formula \eqref{lap-asymp} since the functions $D$ and $W$ are independent of $\mathbf{t}$. We remark that, modulo the exponential term, the result suggests the next order term is of order $3/2$ in $\Delta t$, which coincides with the order in the case for European options, see for example Theorem 2.3 in \cite{ghlow}. 
\begin{theorem}(Discrete monitored Asian option) \label{thm:approx-discrete-Asian} \\
The price $C=C(s_0, 0;K,T)$ of a discretely monitored Asian call struck at $K$ with $K > s_0$ and expiry time $T$ has the following asymptotic expansion as $T \to 0^+$, for fixed $n$,
\bea
  C &=& \E\left[ \left( \frac1n\sum_{i=1}^n S_{t_i} - K \right)^+ \right] \nonumber \\
  &=& \iint\limits_{\left\{\bs:\frac1n \sum_{i=1}^n s_{t_i} \geq K\right\}} \left( \frac1n \sum_{i=1}^n s_{t_i} - K \right) e^{-\frac{D(\bs,\mathbf{t})}{\Delta t}} W(\bs,\mathbf{t}) d\bs \nonumber \\
  &=& \frac{\Delta t^{\frac32}}{\sqrt{2\pi}} \frac{e^{-\frac{D(\bs^*,\mt)}{\Delta t}}}{|\nabla D(\bs^*,\mt)|} \times \left[\frac{\nabla W(\bs^*,\mt)\cdot\nabla D(\bs^*,\mt)}{\sqrt{\det \p_{\mathbf t}^2 D(\bs^*,\mt)}|\nabla D(\bs^*,\mt)|^2} \right. \nonumber \\
  && \qquad \left. + \frac12 \tr\left\{\p^2_{\mt}W(\bs^*,\mt) \left[\p_{\mt}^2 D(\bs^*,\mt)\right]^{-1} \right\} +  \O\left(\Delta t\right) \right], \label{eqn:approx-discrete-Asian}
\eea
where $\bs^*=(s_{t_1}^*,\cdots,s_{t_n}^*)$ is the minimizer of the minimization problem \eqref{eqn:discrete-obj} subject to the constraint $1/n \sum_{i=1}^n s_{t_i} \geq K$.
\end{theorem}

\subsection{Continuously monitored Asian call and the most-likely-path}
\label{subsection: continuous}
The approximate price of an Asian call obtained in Theorem \ref{thm:approx-discrete-Asian} is subject to solving a high-dimensional constrained optimization problem which is daunting in general. However, in the limit as $\Delta t$ approaches zero, the optimization problem converges to a variational problem to which, in certain cases such as Black-Scholes and CIR, the associated Euler-Lagrange equations have closed form solution.
The heuristic computation is given in this section. The rigorous derivation for the leading order using large deviation is done in Section  \ref{sec:LDP}.

Let $\{0=t_0 < t_1 < \cdots < t_n = T\}$ be a partition of the interval $[0,T]$ with $t_i - t_{i-1} = \Delta t := T/n$, for $i=1,\cdots,n$. Then, the price of a continuously monitored Asian call can be written as the limit of the prices of discretely monitored Asian calls as $n \to \infty$. Precisely,
\begin{equation}
\Eof{\left(\frac1T\int_0^T S_t dt - K \right)^+} = \lim_{n\to\infty} \Eof{\left(\frac1n\sum_{i=1}^n S_{t_i} - K\right)^+}
\end{equation}
by applying the Lebesgue dominated convergence theorem.
Hence, to the lowest order, it is natural to approximate the price of a continuously monitored Asian call by taking the limit of the approximate price of discretely monitored Asian call in \eqref{eqn:approx-discrete-Asian} as $\Delta t \to 0$. To be specific, rewrite the logarithm  of \eqref{eqn:approx-discrete-Asian} as
\begin{equation}
\begin{aligned}
 \log C 
& = -\frac{D(\bs^*,\mt)}{\Delta t} - \frac12\log(2\pi) + \frac32\log(\Delta t) - \log\left|\nabla D(\bs^*,\mt) \right| \\ 
& + \log \left[\frac{\nabla W(\bs^*,\mt)\cdot\nabla D(\bs^*,\mt)}{\sqrt{\det \p_{\mathbf t}^2 D(\bs^*,\mt)}|\nabla D(\bs^*,\mt)|^2} + \frac12 \tr\left\{\p^2_{\mt}W(\bs^*,\mt) \left[\p_{\mt}^2 D(\bs^*,\mt)\right]^{-1} \right\} + \O\left(\Delta t\right) \right],
\end{aligned}
\end{equation}
where we recall that $\bs^*$ is the minimizer of \eqref{eqn:discrete-obj}. Notice that the first term in the last expression is dominant as $\Delta t \to 0$.

To determine the limit as $\Delta t \to 0$ of the dominant term, we bring $\Delta t$ back to the objective function \eqref{eqn:discrete-obj} as
\begin{equation}
  \frac1{2\Delta t}\sum_{k=1}^n \left|\varphi(s_{t_i},t_i) - \varphi(s_{t_{i-1}},t_{i-1})\right|^2.
\end{equation}
Note that, since
\begin{equation}
\varphi(s_{t_i},t_i) - \varphi(s_{t_{i-1}},t_{i-1}) = \varphi_s(s_{t_{i-1}},t_{i-1}) \Delta s_{t_i} + \varphi_t(s_{t_{i-1}},t_{i-1}) \Delta t + o(\Delta s_{t_i},\Delta t),
\end{equation}
where $\Delta s_{t_i} = s_{t_i} - s_{t_{i-1}}$, we have
\begin{equation}
\begin{aligned}
 \lim_{\Delta t \to 0} \frac1{2\Delta t}\sum_{k=1}^n \left|\varphi(s_{t_i},t_i) - \varphi(s_{t_{i-1}},t_{i-1})\right|^2  
  &= \lim_{\Delta t \to 0} \frac1{2\Delta t}\sum_{k=1}^n \left| \varphi_s(s_{t_{i-1}},t_{i-1}) \Delta s_{t_i} \right|^2 + o\left((\Delta s_{t_i})^2, \Delta t \right) \\
  &= \lim_{\Delta t \to 0} \frac12\sum_{k=1}^n \left| \frac{\frac{\Delta s_i}{\Delta t}}{a(s_{i-1},t_{i-1})} \right|^2 \Delta t  \\
  &= \frac12 \int_0^T \left[ \frac{\dot s(t)}{a(s(t),t)} \right]^2 dt.
  \end{aligned}
\end{equation}
Therefore, in the limit as $\Delta t$ approaches zero, the optimization problem \eqref{eqn:discrete-obj} becomes the following variational problem\footnote{An equivalent formulation of the same problem in the Black-Scholes case was obtained in \cite{cibelli} (see (1.4)) for the analysis of fundamental solution.} on the space of paths $s: t\mapsto s(t)$
\begin{equation}
  \min_{s:\ t\mapsto s(t)} \frac12 \int_0^T \left[ \frac{\dot s(t)}{a(s(t),t)} \right]^2 dt  \label{eqn:var-obj}
\end{equation}
subject to
\begin{equation}
  \frac1T \int_0^T s(t) dt = K, \quad s(0) = s_0.  \label{eqn:var-cons}
\end{equation}
\begin{defn} $\mbox{ }$ \label{def:mlp} \\
{\rm The optimal path of the variational problem \eqref{eqn:var-obj}:\eqref{eqn:var-cons} is referred to as the {\it most-likely-path} (MLP) for the Asian call struck at $K$.
 }
\end{defn}
In view of the above, one expects that, at the heuristic level and for the leading order, the logarithm of the price $C$ of an out-of-the-money Asian call is approximately given by the solution to the constrained variational problem \eqref{eqn:var-obj}:\eqref{eqn:var-cons}, which is equivalent to determining the most-likely-path.
To determine the MLP, the Euler-Lagrange equation associated with the constrained variational problem \eqref{eqn:var-obj}:\eqref{eqn:var-cons} along with proper boundary conditions is derived as follows.
\begin{lemma}\label{lem:2.2}
The optimal path of the constrained variational problem (\ref{eqn:var-obj}):(\ref{eqn:var-cons}) satisfies the Euler-Lagrange equation
\bea
  && \frac{d}{dt}\left(\frac{\dot s}{a} \right) - \frac{a_t}{a^2} \dot s + \frac\lambda T a = 0  \label{ode}
\eea
with boundary conditions
\bea
  && s(0) = s_0, \quad \dot s(T) = 0, \label{eqn:mlp-bdry-cond}
\eea
where $\lambda$ is chosen such that $1/T\int_0^T s(t) dt = K$.
\end{lemma}
\begin{proof}
We first rewrite the constrained variational problem \eqref{eqn:var-obj}:\eqref{eqn:var-cons} in the Lagrangian form
\begin{equation}
L(s,\lambda) = \frac12 \int_0^T \left[ \frac{\dot s(t)}{a(s(t),t)}\right]^2 dt - \lambda \left(\frac1T\int_0^T s(t) dt - K \right),
\end{equation}
where $\lambda$ is the Lagrange multiplier.
Let $f:[0,T]\mapsto \R$ be a perturbation around the optimal path $s(t)$ with $f(0) = 0$. The first order criterion of optimality yields
\beaa
  0 &=& \left.\frac{d}{d\epsilon}\right|_{\epsilon=0} L(s + \epsilon f,\lambda) \\
  &=& \left.\frac{d}{d\epsilon}\right|_{\epsilon=0} \frac12 \int_0^T \left[ \frac{\dot s + \epsilon \dot f}{a(s+\epsilon f,t)} \right]^2 dt - \lambda \left(\frac1T\int_0^T \{s+\epsilon f\} dt - K \right) \\
  &=& \int_0^T \left[ \frac{\dot s}{a(s,t)} \right] \left[ \frac{\dot f}{a(s,t)} - \frac{a_s(s,t) \dot s f}{a^2(s,t)} \right] dt - \frac\lambda T \int_0^T  f dt \\
  &=& -\int_0^T \left\{\frac{a_s(s,t)}{a^3(s,t)} \dot s^2 + \frac\lambda T \right\}  f dt + \int_0^T \frac{\dot s \dot f}{a^2(s,t)} dt.
\eeaa
By applying integration by parts to the second integral and noting that $ f(0)=0$, the last equality becomes 
\begin{equation}
\begin{aligned}
0 &=
-\int_0^T \left\{\frac{a_s(s,t)}{a^3(s,t)} \dot s^2 + \frac\lambda T \right\}  f dt - \int_0^T \frac{d}{dt}\left[ \frac{\dot s}{a^2(s,t)} \right]  f dt + \frac{\dot s(T)}{a^2(s(T),T)} f(T)\ .
\end{aligned}
\end{equation}
Finally, since $ f$ is arbitrary and $a(s,t) > 0$, we obtain the Euler-Lagrange equation
\beaa
  \frac{a_s(s,t)}{a^3(s,t)} \dot s^2 + \frac\lambda T + \frac{d}{dt}\left[ \frac{\dot s}{a^2(s,t)} \right] = 0
\eeaa
which simplifies to
\beaa
\frac{d}{dt}\left(\frac{\dot s}{a}\right) - \frac{a_t}{a^2} \dot s + \frac\lambda T a = 0
\eeaa
with boundary conditions $s(0) = s_0,  \dot s(T) = 0$.
\end{proof}

We summarize the final result in the following theorem. A rigorous proof based on large deviation theory is postponed to Theorem \ref{thm: ld} in Section \ref{sec:LDP}.
\begin{theorem}(Log price of a continuously monitored Asian call) \label{thm:cts-Asian} \\
The price $C(S_t,t;K,T)$ at time $t$ of a continuously monitored out-of-the-money Asian call struck at $K> s_0$ and expiry time $T$, is approximately equal to
\begin{equation}
\log C(S_t,t;K,T) = -\frac12\int_t^T \left|\frac{\dot s(\tau)}{a(s(\tau),\tau)} \right|^2 d\tau + o(T -t)^{-1} \label{eqn:cts-log-Asian-call}
\end{equation}
where $s:t\mapsto s(t)$ is the solution to the constrained variational problem \eqref{eqn:var-obj}:\eqref{eqn:var-cons}.
\end{theorem}

We conclude the section by deriving closed form expressions for the most-likely-path in the Bachelier, the Black-Scholes, and the Cox-Ingersoll-Ross (CIR) models by solving their corresponding boundary value problems \eqref{ode}:\eqref{eqn:mlp-bdry-cond}.
However, for more general cases we will have to resort to an iteration scheme for numerical computations of the most-likely-path. See Section \ref{sec:numerics} for a numerical scheme.
\begin{example} (Bachelier model) \\
{\rm In the Bachelier model, $a(s,t) = \sigma$, a constant. The Euler-Lagrange equation \eqref{ode} reduces to
\begin{equation}
\frac{d}{dt}\left(\frac{\dot s}{\sigma}\right) + \frac\lambda T \sigma = 0
\end{equation}
whose general solution is $s(t) = -\lambda\sigma^2 t^2/(2T) + c_1 t + c_0$. From the boundary conditions \eqref{eqn:mlp-bdry-cond} together with $1/T\int_0^T s(t) dt = K$, we conclude that $\lambda = 3(K - s_0)/(\sigma^2 T^2)$, $c_1 = 3(K - s_0)/T$, and $c_0 = s_0$. Therefore, the most likely path for Asian call in the Bachelier model is a downward parabola in the $t$-$s$ plane given by
\begin{equation}
\label{eqn: MLP bachelier}
  s(t) = -\frac{3(K - s_0)}2 \left( \frac tT \right)^2 + 3(K - s_0) \frac tT + s_0.
\end{equation}
We remark that the most-likely-path in this case does not depend on $\sigma$.
Also, \eqref{eqn:cts-log-Asian-call} for the Bachelier model becomes
\begin{equation}
  -\frac12\int_0^T \left[ \frac{\dot s(t)}{a(s(t),t)} \right]^2 dt = -\frac{3(K - s_0)^2}{2\sigma^2 T}.
\end{equation}
Thus it captures the exponential decay of an out-of-the-money Asian call in the Bachelier model, see \eqref{eqn:small-time-Bachelier} below.
 }
\end{example}

\begin{example} (Black-Scholes model)\label{ex:bs} \\
{\rm In the Black-Scholes model, $a(s,t) = \sigma s$ with $\sigma$ being a constant. The Euler-Lagrange equation \eqref{ode} reads
\begin{equation}
  -\frac{\ddot s}{s^2} + \frac{\dot s^2}{s^3} = \lambda \sigma^2
\end{equation}
to which the general solutions are
\begin{equation}
  s(t) = \frac{c_1^2}{2\lambda \sigma^2}\left[ 1 - \tanh^2\left( \frac{c_1}2 [c_2 - t] \right) \right]
\end{equation}
and 
\begin{equation}
  s(t) = -\frac{c_1^2}{2\lambda \sigma^2}\left[ 1 + \tan^2\left( \frac{c_1}2 [c_2 - t] \right) \right],
\end{equation}
where $c_1$ and $c_2$ are (to be determined) constants.
The condition $\dot s(T) = 0$ implies that $c_2 = T$ in either case and the conditions $s(0)=s_0$, $\int_0^T s(t)dt = TK$ imply respectively that
\bea
  && \frac{c_1^2}{\lambda \sigma^2(1 + \cosh(c_1 T))} = s_0,  \quad
  c_1 \tanh\left( \frac{c_1 T}2 \right) = \lambda \sigma^2 K T \label{eqn:BS-cond-1}
\eea
and 
\bea
  && \frac{-c_1^2}{\lambda \sigma^2(1 + \cos(c_1 T))} = s_0,  \quad
  - c_1 \tan\left( \frac{c_1 T}2 \right) = \lambda \sigma^2 K T. \label{eqn:BS-cond-2}
\eea
By dividing the two equations in \eqref{eqn:BS-cond-1} and rearranging terms, $c_1$ is given by the solution to the equation
\begin{equation}
  \frac{c_1 T}{\sinh(c_1 T)} = \frac{s_0}K \quad \mbox{ if } s_0 < K 
\end{equation}
or $c_1 = f^{-1}\left( s_0/K \right)/T$, where $f(x) = x/\sinh x$. On the other hand, from \eqref{eqn:BS-cond-2} we obtain that 
\begin{equation}
  \frac{\sin(c_1 T)}{c_1 T} = \frac K{s_0} \quad \mbox{ if } K > s_0
\end{equation}
or $c_1 = g^{-1}\left( K/s_0 \right)/T$, where $g(x) = \sin x/x$ for $x \in [0,\pi]$.
 
Hence, the most-likely-path $s(t)$ for $0 \leq t \leq T$ in the Black-Scholes model after simplification becomes
\begin{equation}
\label{eqn: MLP BS}
s(t) = \frac{s_0 \cosh^2\left(\frac{c_1 T}2\right)}{\cosh^2\left(\frac{c_1 (T-t)}2\right)},
\end{equation}
where $c_1 = f^{-1}\left( s_0/K \right)/T$ for $s_0 < K$; whereas for $s_0 > K$
\begin{equation}
s(t) = \frac{s_0 \cos^2\left(\frac{c_1 T}2\right)}{\cos^2\left(\frac{c_1 (T-t)}2\right)}
\end{equation}
and $c_1 = g^{-1}\left( K/s_0 \right)/T$. Moreover, for $s_0 < K$,
\begin{equation}
  \frac12 \int_0^T \left| \frac{\dot s(t)}{\sigma s(t)} \right|^2 dt = \frac{f^{-1}\left( \frac{s_0}K \right)}{\sigma^2 T}\left[ \frac12 f^{-1}\left( \frac{s_0}K \right) - \tanh\left( \frac12 f^{-1}\left( \frac{s_0}K \right) \right) \right],
\end{equation}
and for $s_0 > K$
\begin{equation}
  \frac12 \int_0^T \left| \frac{\dot s(t)}{\sigma s(t)} \right|^2 dt = \frac{g^{-1}\left( \frac K{s_0} \right)}{\sigma^2 T}\left[ \tan\left( \frac12 g^{-1}\left( \frac K{s_0} \right) \right) - \frac12 g^{-1}\left( \frac K{s_0} \right) \right].
\end{equation}
 }
\end{example}
\begin{example} (Cox-Ingersoll-Ross model)\label{ex:cir} \\
{\rm For the CIR model, $a(s,t) = \sigma \sqrt{s}$, where $\sigma$ is a constant. The Euler-Lagrange equation \eqref{ode} becomes
\begin{equation}
  \frac{d}{dt}\left(\frac{\dot s(t)}{\sqrt{s(t)}} \right)+2 \, \phi \, \sqrt{s(t)}
= 2\, \frac{d^2}{dt^2}\left({\sqrt{s(t)}} \right) + 2 \,\phi\, \sqrt{s(t)} = 0
\end{equation}
for some constant $\phi$ with general solution given by
\begin{equation}
\label{eqn: MLP CIR}
  \sqrt{s(t) }= c_2 \,\cos\left(\sqrt\phi\,(c_1-t)\right)
\end{equation} 
if $\phi > 0$ and 
\begin{equation}
\label{eqn: MLP CIR-2}
  \sqrt{s(t) }= c_2 \,\cosh\left(\sqrt{-\phi}\,(c_1-t)\right)
\end{equation}
if $\phi < 0$.
The boundary conditions \eqref{eqn:mlp-bdry-cond} imply respectively that
\beaa
  &&  c_1= T
  \quad \mbox{ and } \quad c_2= \frac{\sqrt{s_0}}{\cos\left(\sqrt \phi\,T\right)} \quad \mbox{ or } \quad
  \frac{\sqrt{s_0}}{\cosh\left(\sqrt{-\phi}\,T\right)}.
\eeaa
Thus
\begin{equation}
s(t) = s_0 \,\left\{\frac{\cos\left(\sqrt\phi\,(T-t)\right)}{ \cos\left(\sqrt\phi\,T\right)  }\right\}^2 \quad \mbox{ or } \quad 
s(t) = s_0 \,\left\{\frac{\cosh\left(\sqrt{-\phi}\,(T-t)\right)}{ \cosh\left(\sqrt{-\phi}\,T\right)  }\right\}^2.  
\end{equation}
The parameter $\phi$ is determined by the solution to the equation
\begin{equation}
 \frac1T\, \int_0^T \,s(t)\, dt = K = \frac{s_0}{2\, \sqrt\phi\,T}
 \left[\tan\left(\sqrt\phi\,T \right) + \sqrt\phi\,T  \sec^2\left( \sqrt\phi\,T \right)\right]
\end{equation}
if $s_0 < K$ and by 
\begin{equation}
 \frac K{s_0} = \frac1{2\, \sqrt{-\phi}\,T}
 \left[\tanh\left(\sqrt{-\phi}\,T \right) + \sqrt{-\phi}\,T  \sech^2\left( \sqrt{-\phi}\,T \right)\right]
\end{equation}
if $s_0 > K$.
Finally, we have, subject to the determination of $\phi$, that 
\begin{equation}
  \int_0^T \left|\frac{\dot s(t)}{\sigma \sqrt{s(t)}}\right|^2 dt = \frac{s_0\, \sqrt\phi}{\sigma^2} \, \left[2\,\sqrt\phi\,T -\sin \left(2 \,\sqrt\phi\,T \right) \right] \,\sec^2\left(\sqrt\phi\,T\right)
  \end{equation}
for $s_0 < K$ and 
\begin{equation}
  \int_0^T \left|\frac{\dot s(t)}{\sigma \sqrt{s(t)}}\right|^2 dt = \frac{s_0\, \sqrt{-\phi}}{\sigma^2} \, \left[ -2\,\sqrt{-\phi} \,T  + \sinh\left(2 \,\sqrt{-\phi} \,T \right) \right] \,\sech^2\left( \sqrt{-\phi} \, T \right)
 \end{equation}
if $s_0 > K$.  
 }
\end{example}

\subsection{Implied normal volatility for Asian options}
As far as implied volatility is concerned, we opt to use the Bachelier model as benchmark rather than the conventional Black-Scholes model because of the lack of analytical expression for Asian options in the Black-Scholes model. Such defined implied volatility in the European counterpart is sometimes used and referred to as the {\it implied normal volatility} in practice. For European calls, this approximation is good whenever $\sigma\sqrt{T}$ is small, at least for at-the-money options, see for example \cite{ST07}. We expect the same should hold for Asian options, though the lack of analyticity makes it hard to check. In particular, note that the approximation might be problematic for small strike prices as the ratio $S/K$ will be large.

Recall that, by straightforward calculations, under the Bachelier model
\begin{equation}
dS_t = \sigma_b dW_t,
\end{equation}
where $\sigma_b$ is a constant, the price of a continuously monitored Asian call option struck at $K$ with expiry $T$ has the closed form expression
\bea
C_b(K,T,\sigma_b) = \frac{\sigma_b \sqrt T}{\sqrt{6\pi}}e^{-\frac{3(s_0 - K)^2}{2\sigma_b^2 T}} + (s_0 - K) N\left( \frac{\sqrt3(s_0 - K)}{\sigma_b\sqrt T} \right)
\label{eq:BachelierAsianCall}
\eea
since the average price $1/T \int_0^T S_tdt$ is normally distributed with mean $s_0$ and variance $\sigma_b^2 T/3$. In \eqref{eq:BachelierAsianCall}, $N(\cdot)$ denotes the cumulative distribution function of standard normal distribution. On the other hand, for a discretely monitored Asian call struck at $K$ with expiry $T$ in the Bachelier model, since $S_{t_1} + \cdots + S_{t_n}$ is normally distributed with
\beaa
  && \E[S_{t_1} + \cdots + S_{t_n}] = n s_0,  \\
  && \var[S_{t_1} + \cdots + S_{t_n}] = \sigma_b^2 T \frac{(n+1)(2n+1)}6,
\eeaa
its price $C_b^d$ is given by
\bea
C_b^d(K,T,\sigma_b) = \frac{\sigma_b \sqrt T}{\sqrt{2 A_n \pi}}e^{-\frac{A_n(s_0 - K)^2}{2\sigma_b^2 T}} + (s_0 - K) N\left( \frac{\sqrt{A_n}(s_0 - K)}{\sigma_b\sqrt T} \right),
\label{eq:BachelierAsianCall-d}
\eea
where $A_n := 6n^2/(n+1)(2n+1)$. Apparently, $A_n \to 3$ as $n\to\infty$.

Regarding the small time expansion, we remark that by using the asymptotic expansion $N(x) = N'(x)[-1/x + 1/x^3 + \mathcal O(1/x^5)]$ as $x \to -\infty$, we have
\bea
C_b(s_t,K) &=& \frac{\sigma_b \sqrt{T-t}}{\sqrt{6\pi}}e^{-\frac{3(s_t - K)^2}{2\sigma_b^2(T-t)}} + (s_t - K) N\left( \frac{\sqrt3(s_t - K)}{\sigma_b\sqrt{T-t}} \right) \nonumber \\
&=& e^{-\frac{3(s_t - K)^2}{2\sigma_b^2(T-t)}} \left[\frac{\sigma_b^3 (T-t)^{\frac32}}{3\sqrt{6\pi}(K - s_0)^2} + \mathcal O(T-t)^{\frac52}\right] \label{eqn:small-time-Bachelier}
\eea
as $t \to T^-$. Similarly, in the discrete case as $t \to T$,
\bea
C_b^d(s_t,K) &=& e^{-\frac{A_n(s_t - K)^2}{2\sigma_b^2(T-t)}} \left[\frac{\sigma_b^3 (T-t)^{\frac32}}{A_n\sqrt{2 A_n \pi}(K - s_0)^2} + \mathcal O(T-t)^{\frac52}\right]. \label{eqn:small-time-Bachelier-d}
\eea

We remark that once the small time asymptotic for the price of an out-of-the-money call is established on the model side, it is a common practice to derive the small time asymptotic of implied volatility thereby. 
To that end, the following expansion for the implied normal volatility in terms of the call price given in \cite{grunspan} will be helpful.    
\begin{proposition}
For a fixed strike out-of-the-money call with time to expiry $T$, let $C = C(T)$ be the price of the European call regarded as a function of $T$. Then as time to expiry $T \to 0$, the implied normal volatility $\sigma_N$ has the asymptotic
\bea 
&& \sigma_N^2 T = \frac{(s_0 - K)^2}{2(\log s_0 - \log C(T))} + o(\log C(T)). \label{eqn:normal-vol-exp} 
\eea
as $T \to 0$.
\end{proposition}

\subsubsection{\textbf{Implied normal volatility for discretely monitored Asian option}}
The implied normal volatility for a discretely monitored Asian call is defined by solving the following equation for $\sigma_b$
\begin{equation}
C^d(s_t,K,T) = \frac{\sigma_b \sqrt T}{\sqrt{2 A_n \pi}}e^{-\frac{A_n(s_0 - K)^2}{2\sigma_b^2 T}} + (s_0 - K) N\left( \frac{\sqrt{A_n}(s_0 - K)}{\sigma_b\sqrt T} \right), \label{eqn:imp-vol-discrete}
\end{equation}
where $A_n := 6n^2/(n+1)(2n+1)$ and $C^d$ is the price of a discretely monitored Asian call obtained by the model or from the market. Obviously, among other parameters, such defined $\sigma_b$ depends on $K$ and $T$.

To derive an asymptotic expansion for the implied normal volatility $\sigma_b$ defined in \eqref{eqn:imp-vol-discrete} in small time, the idea, as in \cite{ghlow}, is to compare the corresponding terms in the expansions on both side of \eqref{eqn:imp-vol-discrete}. The lowest order term is thus obtained by matching the exponential terms on both side of \eqref{eqn:imp-vol-discrete}. Precisely, recall the small time asymptotic of the price of a discrete Asian call from Theorem \ref{thm:approx-discrete-Asian}
\begin{equation}
\begin{aligned}
 \E\left[ \left( \frac1n\sum_{i=1}^n S_{t_i} - K \right)^+ \right] 
&= \frac{\Delta t^{\frac32}}{\sqrt{2\pi}} \frac{e^{-\frac{D(\bs^*,\mt)}{\Delta t}}}{|\nabla D(\bs^*,\mt)|} \times \left[\frac{\nabla W(\bs^*,\mt)\cdot\nabla D(\bs^*,\mt)}{\sqrt{\det \p_{\mathbf t}^2 D(\bs^*,\mt)}|\nabla D(\bs^*,\mt)|^2} \right. \nonumber \\
& \qquad \left. + \frac12 \tr\left\{\p^2_{\mt}W(\bs^*,\mt) \left[\p_{\mt}^2 D(\bs^*,\mt)\right]^{-1} \right\} +  \O\left(\Delta t\right) \right],
\end{aligned}
\end{equation}
where $\Delta t = T/n$. By matching the exponential terms in the above expression and in \eqref{eqn:small-time-Bachelier-d}, we have
\begin{equation}
\begin{aligned}
& e^{-\frac{A_n(s_t - K)^2}{2\sigma_b^2 T}} = e^{-\frac{D(\bs^*,\mathbf{t})}{\Delta t}} \\
&\Longrightarrow \frac{A_n(s_t - K)^2}{2\sigma_b^2 T} = \frac{D(\bs^*,\mathbf{t})}{\Delta t} = \frac{n D(\bs^*,\mathbf{t})}T \\
&\Longrightarrow \frac1{\sigma_b^2} = \frac{2 n D(\bs^*,\mathbf{t})}{A_n(s_t - K)^2} = \frac n{A_n(s_t - K)^2}\sum_{k=1}^n \left|\varphi(s^*_{t_i},t_i) - \varphi(s^*_{t_{i-1}},t_{i-1})\right|^2.
\end{aligned}
\end{equation}
Hence, the lowest (zeroth) order approximation of the implied normal volatility is given by
\begin{equation}
\sigma_b = \sqrt{A_n} |s_t - K| \left(n \sum_{k=1}^n \left|\varphi(s^*_{t_i},t_i) - \varphi(s^*_{t_{i-1}},t_{i-1})\right|^2 \right)^{-\frac12} + o(1),
\end{equation}
where recall that $(s_{t_1}^*,\cdots,s_{t_n}^*)$ is the solution to the $n$-dimensional constrained optimization problem \eqref{eqn:discrete-obj}:\eqref{eqn:discrete-constr}. We summarize the result in the following theorem.
\begin{theorem}(Implied normal volatility asymptotic for discrete Asian call) \label{thm:imp-vol-approx-discrete} \\
For a discretely monitored out-of-the-money Asian call struck at $K$, i.e., $s_0 < K$, in which the underlying is driven by the local volatility model \eqref{eqn:local-vol}, the implied normal volatility $\sigma_b^d$ defined as in \eqref{eqn:imp-vol-discrete} has the asymptotic expansion as $T \to 0^+$
\begin{equation}
\sigma_b^d(K,T) = \sigma^d_{b,0}(K) + \mathcal O(T),
\end{equation}
where
\begin{equation}
\sigma^d_{b,0}(K) = \sqrt{A_n} |s_t - K| \left(n \sum_{k=1}^n \left|\varphi(s^*_{t_i},t_i) - \varphi(s^*_{t_{i-1}},t_{i-1})\right|^2 \right)^{-\frac12},
\end{equation}
with $A_n = 6n^2/(n+1)(2n+1)$ and $(s_{t_1}^*,\cdots,s_{t_n}^*)$ the solution to the constrained optimization problem \eqref{eqn:discrete-obj}:\eqref{eqn:discrete-constr}.
\end{theorem}

\subsubsection{\textbf{Implied normal volatility for continuously monitored Asian option}}
For continuously monitored Asian calls, the implied normal volatility is thus defined by solving the following equation for $\sigma_b$
\begin{equation}
C(s_t,K,T) = \frac{\sigma_b \sqrt{T-t}}{\sqrt{6\pi}}e^{-\frac{3(s_t - K)^2}{2\sigma_b^2(T-t)}} + (s_t - K) N\left( \frac{\sqrt3(s_t - K)}{\sigma_b\sqrt{T-t}} \right), \label{eqn:imp-vol-cts}
\end{equation}
where $C$ is the price of an Asian call obtained by the model or from the market. Notice that the right hand side of \eqref{eqn:imp-vol-cts} can be regarded as the price of a European call option in the Bachelier model but with one third of $\sigma_b^2$ as the variance (volatility squared) parameter. In other words, in the Bachelier world, the price of an Asian call is equal to the price of its European counterpart but with only one third of variance. By combining the asymptotics \eqref{eqn:cts-log-Asian-call} for the call price and \eqref{eqn:normal-vol-exp} for the implied normal volatility, a small time asymptotic for the implied normal volatility of an out-of-the-money Asian call is established. 
Alternatively, we may also obtain the same approximation by straightforwardly taking the limit as $n\to\infty$ of $\sigma^d_{b,0}$ in Theorem \ref{thm:imp-vol-approx-discrete}. We summarize the result in the following theorem.
\begin{theorem}(Implied normal volatility asymptotic for continuous Asian call) \label{thm:imp-vol-approx} \\
For a continuously monitored out-of-the-money Asian call struck at $K$, i.e., $s_0 < K$, in which the underlying is driven by the local volatility model \eqref{eqn:local-vol}, the implied normal volatility $\sigma_b$ defined in \eqref{eqn:imp-vol-cts} has the asymptotic expansion as $T \to 0^+$
\begin{equation}
\sigma_b(K,T) = \sigma_{b,0} + o(T),
\end{equation}
where
\begin{equation}
\sigma_{b,0} = \left( \frac T{3(K - s_0)^2} \int_0^T \left[ \frac{\dot{\tilde s}(t)}{a(\tilde s(t),t)}\right]^2 dt \right)^{-\frac12},
\end{equation}
and $\tilde s(t)$ is the most-likely-path for Asian option determined by solving the variational problem \eqref{eqn:var-obj}:\eqref{eqn:var-cons}.
\end{theorem}
\begin{proof}
Recall from \eqref{eqn:normal-vol-exp} that the implied normal volatility $\sigma$ of a European option has the asymptotic
\beaa 
&& \sigma^2 T = \frac{(s_0 - K)^2}{2(\log s_0 - \log C(T))} + o(\log C(T)).
\eeaa
Thus, by substituting $\log C(T)$ with \eqref{eqn:cts-log-Asian-call} and using the fact that Asian variance equals one third of its European counterpart in Bachelier world, we obtain for implied normal volatility $\sigma_b$ of an Asian call
\beaa
\frac{\sigma_b^2 T}3 &=& \frac{(s_0 - K)^2}{2(\log s_0 - \log C(T))} + o(\log C(T)) \\
&=& -\frac{(K - s_0)^2}{2\log C(T)} + o(\log C(T)) \\
&=& (K - s_0)^2\left[\int_0^T \left[\frac{\dot{\tilde s}(t)}{a(\tilde s(t), t)}\right]^2 dt \right]^{-1} + o(T)
\eeaa
where we used \eqref{eqn:cts-log-Asian-call} in the last equality. Finally, the result is obtained by rearranging terms.
\end{proof}

\begin{example}{(Time-dependent Bachelier model)}
{\rm
Finally, we consider a time-dependent Bachelier model in which 
\begin{equation}
  d S_t = \sigma\,\theta(t)\,dW_t.
\end{equation}
Note that in this case $a(s,t) = \sigma \theta(t)$, the Euler-Lagrange equation \eqref{ode} reduces to 
\beaa
&& \frac{d}{dt}\left(\frac{\dot s}{\sigma\theta(t)}\right) - \frac{\theta'(t) \dot s}{\sigma \theta^2(t)} + \frac\lambda T \sigma \theta(t) = 0 \\
&\Longrightarrow& \frac1{\theta(t)} \frac{d}{dt}\left(\frac{\dot s}{\theta(t)}\right) - \frac{\theta'(t)}{\theta^2(t)}\,\frac{\dot s}{\theta(t)} + \frac\lambda T \sigma^2 = 0
\eeaa
Integrating the last equation subject to the condition $\dot s(T)=0$ gives
\bea
 \frac{\dot s(t)}{\theta^2(t)} = \frac\lambda T \sigma^2(T-t) \label{eqn:sdot-td-B}
\eea
for some constant $\lambda$ set using the condition
\begin{equation}
\frac{1}{T}\,\int_0^T\,s(t)\,dt = K.
\end{equation}
This gives
\begin{equation}
\lambda = \frac{T^2(K-s_0)}{\sigma^2\,\int_0^T (T-u)^2 \theta^2(u)du}.
\end{equation}
Also, from \eqref{eqn:sdot-td-B} one easily obtain
\begin{equation}
  \int_0^T \left|\frac{\dot s(t)}{a(s,t)}\right|^2 dt = \int_0^T \left|\frac{\dot s(t)}{\sigma\theta(t)}\right|^2 dt 
  = \frac{\lambda^2 \sigma^2}{T^2}\,\int_0^T\,(T-t)^2 \theta^2(t)\,dt = \frac{T^2(K-s_0)^2}{\sigma^2 \int_0^T (T-t)^2\theta^2(t) dt}.
  \end{equation}
This gives
\beq
\sigma_{b,0}^2 = \frac{3\,\sigma^2}{T^3}\,\int_0^T\,(T - t)^2\,\theta^2(t)\,dt.
\label{eq:tdBachelier}
\eeq
On the other hand, we have $\bar S_T = 1/T \,\int_0^TS_t\,dt$ where
 \begin{equation}
S_t =s_0 +\sigma\, \int_0^t\,\theta(u)\,dW_u.
\end{equation}
A straightforward computation then yields
\begin{equation}
\var\left[ \bar S_T \right] = \frac{\sigma^2}{T^2}\,\int_0^T\,(T-u)^2\,\theta(u)^2\,du
\end{equation}
so the {\em exact} Bachelier implied volatility in the time-dependent Bachelier case is given by
\begin{equation}
\sigma_b^2 = \frac{3\,\sigma^2}{T^3}\,\int_0^T\,(T-u)^2\,\theta(u)^2\,du.
\end{equation}
Thus, the most-likely-path approximation is exact in the time-dependent Bachelier case.
}
\end{example}


\subsection{Approximation of Greeks}
The approximation of implied volatility in Theorem \ref{thm:imp-vol-approx} is also applicable for approximations of Greeks. For example, we may calculate the delta as follows. For notational simplicity, denote by $v_b := \sigma_b^2 T$. Suppressing and holding the other parameters fixed, since the Bachelier implied volatility is defined through
\beaa
C(s) = C_b(s, v_b), 
\eeaa
the delta $\Delta$ satisfies
\beaa
\Delta := \frac{\p C}{\p s} = \frac{\p C_b}{\p s} + \frac{\p C_b}{\p v_b} \frac{\p v_b}{\p s},
\eeaa
where $C_b$ is the function defined in \eqref{eqn:imp-vol-cts}.
Note that by straightforward calculations we have 
\beaa
&& \frac{\p C_b}{\p s} = N\left(\frac{\sqrt3 (s - K)}{\sqrt{v_b}}\right) \quad \mbox{ and } \quad 
 \frac{\p C_b}{\p v_b} = \frac1{2\sqrt{6\pi v_b}} e^{-\frac{3(s - K)^2}{2v_b}}.
\eeaa
Thus, with $v_b \approx v_{b,0} := \sigma_{b,0}^2 T$, it follows that   
\begin{equation}
\Delta \approx N\left(\frac{\sqrt3 (s - K)}{\sqrt{v_{b,0}}}\right) + \frac1{2\sqrt{6\pi v_{b,0}}} e^{-\frac{3(s - K)^2}{2v_{b,0}}} \frac{\p v_{b,0}}{\p s}. \label{eqn:approx-delta}
\end{equation}
The expressions on the right hand side of \eqref{eqn:approx-delta} can be calculated easily except the last term which can be calculated as follows. Recall that 
\begin{equation}
v_{b,0} = 3(K - s)^2 \left[\int_0^T \left[\frac{\dot{\tilde s}(t)}{a(\tilde s(t), t)}\right]^2 dt \right]^{-1},
\end{equation} 
we have
\bea
\frac{\p v_{b,0}}{\p s} &=& 6(s - K) \left[\int_0^T \left[\frac{\dot{\tilde s}(t)}{a(\tilde s(t), t)}\right]^2 dt \right]^{-1} \nonumber \\
&& - 3(K - s)^2 \left[\int_0^T \left[\frac{\dot{\tilde s}(t)}{a(\tilde s(t), t)}\right]^2 dt \right]^{-2} \, \frac{\p}{\p s} \int_0^T \left[\frac{\dot{\tilde s}(t)}{a(\tilde s(t), t)}\right]^2 dt \nonumber \\
&=& \frac{2 v_{b,0}}{s - K} - \frac{v_{b,0}^2}{3(K - s)^2} \, \frac{\p}{\p s} \int_0^T \left[\frac{\dot{\tilde s}(t)}{a(\tilde s(t), t)}\right]^2 dt. \label{eqn:vb0-ps}
\eea 
The integral in \eqref{eqn:vb0-ps} in the general case needs to be evaluated numerically. 
Similarly, since the gamma $\Gamma$ satisfies
\beaa
&& \Gamma := \frac{\p \Delta}{\p s} = \frac{\p^2 C_b}{\p s^2} + 2 \frac{\p^2 C_b}{\p s \p v_b} \frac{\p v_b}{\p s} + \frac{\p^2 C_b}{\p v_b^2} \left(\frac{\p v_b}{\p s}\right)^2 + \frac{\p C_b}{\p v_b} \frac{\p^2 v_b}{\p s^2},
\eeaa
an approximation of $\Gamma$ is given by substituting $v_b$ with $v_{b,0}$, subject to numerically evaluations of the derivatives $\p_s v_{b,0}$ and $\p_s^2 v_{b,0}$. 
However, we remark that in the Black-Scholes (Example \ref{ex:bs}) and the CIR (Example \ref{ex:cir}) models, since the $\sigma_{b,0}$'s have closed form expressions, closed expressions for Greeks are available. 
We refer the reader to \cite{pirjol-zhu-1} for more detailed discussions on Greeks of Asian options in the Black-Scholes model.

%
%

\section{Large deviation principle} \label{sec:LDP}
In this section, we prove Theorem \ref{thm: ld} which is a large deviation reformulation of Theorem \ref{thm:cts-Asian} for continuously monitored options.
Finer tools would be needed to go beyond the leading order of the asymptotic expansion in Theorem \ref{thm:approx-discrete-Asian} suggested by the discrete approximation.
However, in particular cases, these terms are given  by a Girsanov change of measure.
We will illustrate this using the Bachelier model at the end of the section.
A more refined expansion for more general processes does not appear tractable with current methods.
\begin{theorem}[Large deviation for the log-price of continuous monitored Asian call]
\label{thm: ld}
The price $C(s_0,0;K,T)$ at time $t=0$ of a continuously monitored out-of-the-money Asian call struck at $K> s_0$ with expiry time $T$
for the price process $(S_t,t\geq 0)$ of \eqref{eqn:local-vol}
admits the following expansion in $T$
\begin{equation}
\label{eqn: first order}
 C(s_0,0;K,T)
 = \exp \left\{ -\frac{J(K)}{T} + o\left(T^{-1}\right) \right\}
\end{equation}
 where
\begin{equation}
 J(K)=\inf\left\{I(f): f\in \mathcal C([0,T]) \text{ and } \int_0^T f(u)=K \right\}\ .
 \end{equation}
 Moreover, if $f(t)=x+\int_0^t a(f(s),s)g(s)ds$ for some $g\in L^2([0,T])$ then
 \begin{equation}
 \label{eqn: rate fct}
I(f)= \frac{1}{2}\int_0^T \left(\frac{\dot{f}(t)}{a(f(t),0)}\right)^2 dt \ ,
 \end{equation}
 and $I(f)$ is $\infty$ otherwise.
\end{theorem}

The idea to prove the theorem is to treat $\e=T$ as a small parameter and  expand around $\e=0$.
We do a simple time-change $t=u T$ and write for simplicity
\begin{equation}
S^{\e}:=(S^{\e}_{u}, u\in [0,1])=(S_{uT} , u\in [0,1])\ .
\end{equation}
Note that with this notation:
\begin{equation}
\frac{1}{T}\int_0^T S_t dt= \int_0^1 S^\e_{u} \ du\ .
\end{equation}
By the scaling property of Brownian motion, the process $(S^\e_{u}, u\in [0,1])$ satisfies the SDE
\begin{equation}
\label{eqn: family}
dS^\e_{u}= a_\e(S^\e_{u},u) \ \sqrt{\e} dW_u\ ,
\end{equation}
where $a_\e(x,u) := a(x,u\e)$. Since $a(x,\cdot)$ is assumed to be continuous uniformly in $x$, we can write
\begin{equation}
a(x,u\e)=a(x,0) + \mathcal O(\e)\ .
\end{equation}

A {\it large deviation principle} for the process means that
there exists a  {\it rate function} $I:\mathcal C([0,T])\to [0,\infty]$ that is lower semi-continuous and has compact level sets such that
 for any Borel subset $A$ of paths in $\mathcal C([0,1])$,
 \begin{equation}
 \label{eqn: limits ld}
 -\inf_{s\in \text{int(A)}} I(s)\leq \liminf_{\e\to 0} \e \log \PP(S^\e\in A)\leq \limsup_{\e\to 0}\e \log \PP(S^\e\in A)\leq -\inf_{s\in \text{cl(A)}} I(s)
 \end{equation}
 where $\text{int($A$)}$ denotes the interior of $A$ and $\text{cl($A$)}$ is its closure.
 Roughly speaking, a large deviation principle quantifies the probability of atypical path at the exponential scale with the help of the rate function.
The proof of Theorem \ref{thm: ld} is based on an extension of the Freidlin-Wentzell theorem, see e.g. \cite{dembo-zeitouni}.
The standard statement of the theorem holds for time-homogeneous drift and volatility under Lipschitz and boundedness assumptions.
 Here we will use a recent result of \cite{chiarini-fischer}.
Note that the locally Lipschitz condition \eqref{eqn: lipschitz} is not satisfied in the CIR model.
However, in this case, the volatility is time-homogeneous and a weaker assumption is needed for a large deviation to hold.
This is the content of Theorem 4 in \cite{chiarini-fischer}.

\begin{theorem}[Theorem 2, Theorem 4 and Example 1 in \cite{chiarini-fischer}]
\label{thm: ldp diffusion}
The family of diffusions \eqref{eqn: family} satisfy a large deviation principle with rate function
 \begin{equation}
I(f)= \inf_{\{g\in L^2([0,1]): f(t)=x+\int_0^t a(f(s),0)g(s) ds\}} \frac{1}{2}\int_0^T |g(t)|^2 dt\ ,
 \end{equation}
 whenever the set $\{g\in L^2([0,1]): f(t)=x+\int_0^t a(f(s),0)g(s) ds\}$ is non-empty and $I(f)$ is $\infty$ otherwise.
 \end{theorem}

 Formally, it is good to think of the set $\{g\in L^2([0,T]): f(t)=x+\int_0^t a(f(s),0)g(s) ds\}$ as the set of ``white-noise paths" $t\mapsto \dot{W}_t$.
 The map $g\mapsto f$ where $f$ is the solution of $f(t)=x+\int_0^t \bar{a}(f(s),s)g(s) ds$ can then be thought of as the map sending an underlying Brownian path
 to the corresponding diffusion path. In the case where the volatility is non-zero for the path $f$, the map can be inverted and the rate function reduces to the
 simplest case
 \begin{equation}
 \label{eqn: I simple}
 I(f)= \frac{1}{2}\int_0^1 \left(\frac{\dot{f}(t)}{a(f(t),0)}\right)^2 dt\ .
 \end{equation}
This is certainly the case for geometric Brownian motion and the CIR model.
The reader is referred to Theorem 1 in \cite{chiarini-fischer} and Proposition 3.11 in \cite{baldi-caramellino} for general sufficient conditions for \eqref{eqn: I simple} to hold.

\begin{proof}[Proof of Theorem \ref{thm: ld}]
Observe that for any random variable $X$ and $K>0$, we have the identity
\begin{equation}
\label{eqn: expect}
\E[(X-K)^+]=\int_K^\infty \PP(X>x) \ dx.
\end{equation}
For simplicity, write $T$ for the functional $T: \mathcal C([0,1])\to \R$ with $T(f)= \int_0^1 f(u) du$.
Note that $T$ is continuous on  $\mathcal C([0,1])$ equipped with the topology of uniform convergence.
By the contraction principle (see e.g. \cite{dembo-zeitouni}) and Theorem \ref{thm: ldp diffusion}, the family of random variables $( T(S^\e), \e >0)$ satisfies a large deviation principle with rate function $J: \R\to [0,\infty]$ with
\begin{equation}
J(y)=\inf\{ I(f): T(f) =y \ , f\in \mathcal{C} ([0,1])\}\ , y\in \R\ .
\end{equation}
By \eqref{eqn: expect}, we have
\begin{equation}
\log \E\left[\left(T(S^\e) -K\right)^+\right]
=\log \int_K^\infty \PP(T(S^\e)>x) \ dx.
\end{equation}
On one hand, we have for any $M>K$
\begin{equation}
\begin{aligned}
&\log \int_K^\infty \PP(T(S^\e)>x) \ dx\\
&= \left(\log  \int_K^M \PP(T(S^\e)>x) \ dx\right) + \log \left(1+ \frac{\log  \int_M^\infty \PP(T(S^\e)>x) \ dx}{\log  \int_K^M \PP(T(S^\e)>x) \ dx}\right)\ .
\end{aligned}
\end{equation}
For $M$ large enough (possibly dependent on $\e$), the term in the second parenthesis is smaller than $2$.
Pick $M=\e^{-1}$. For $\e$ small enough we thus have
\begin{equation}
\log \int_K^\infty \PP(T(S^\e)>x) \ dx
\leq \log 2 + \log (\e^{-1}-K)+\log \PP(T(S^\e)>K)\ .
\end{equation}
This proves that
\begin{equation}
\limsup_{\e\to 0}\e \log \int_K^\infty \PP(T(S^\e)>x) \ dx
\leq \lim_{\e\to 0} \e \log \PP(T(S^\e)>K)=-J(K)\ ,
\end{equation}
since $(T(S^\e), \e>0)$ satisfies a large deviation with rate function $J$.

On the other hand, for $\delta>0$
\begin{equation}
\int_K^\infty \PP(T(S^\e)>x) \ dx
\geq  \int_K^{K+\delta} \PP(T(S^\e)>x) \ dx
\geq \delta \PP(T(S^\e)>K+\delta)\ .
\end{equation}
Therefore, for any $\delta >0$, the above with the use of \eqref{eqn: limits ld} becomes
\begin{equation}
\liminf_{\e\to 0}\e \log \int_K^\infty \PP(T(S^\e)>x) \ dx
\geq
-J(K+\delta)\ .
\end{equation}
It remains to show that $\lim_{\delta\to 0}J(K+\delta)=J(K)$.
First, notice that $\liminf_{\delta \to 0} J(K+\delta)\geq J(K)$ since $J$ is lower semi-continuous (it is a rate function).
So it suffices to show that  $\limsup_{\delta \to 0} J(K+\delta)\leq J(K)$.
By definition, $J(K+\delta)=\inf\{I(x):x\in \mathcal{C}([0,1]) \text{ and } \int_0^1 x(u)du=K+\delta\}$.
Pick a sequence $(y_n)\in \mathcal C([0,1]) $ such that $I(y_n)\to J(K)$ and $\int_0^1 y_n(u)du =K$.
By definition of the infimum, this sequence can be picked such that $I(y_n)<J(K)-1/n$.
Pick $z$ a differentiable function on $[0,1]$ such that $\int z=1$ and $z(0)=0$.
Then $\int (y_n(u)+\delta z(u))du=K+\delta$. Moreover
\begin{equation}
J(K+\delta)< I(y_n+\delta z)\ .
\end{equation}
It is easy to check that for fixed $n$, $\lim_{\delta\to 0}I(y_n+\delta z)=I(y_n)$. Therefore
\begin{equation}
\limsup_{\delta \to 0}J(K+\delta)< I(y_n)<J(K)-1/n\ .
\end{equation}
Since $n$ is arbitrary, this proves the claim.
\end{proof}

The asymptotic expansion \eqref{eqn:cts-log-Asian-call}
suggests that the lower order corrections to \eqref{eqn: first order} are much smaller than the dominant term in $T^{-1}$.
We expect more precisely that
\begin{equation}
\label{eqn: lower order}
 \E\left[\left(\frac{1}{T}\int_0^T S_{u}\ du -K\right)^+\right]
 =\exp\left\{-\frac{J(K)}T  + \frac{3}{2}\log T +\mathcal O(1)\right\}\ .
 \end{equation}
 This can be verified rigorously for the Bachelier model (Example 1) using a Girsanov change of measure designed to tilt towards the {\it most-likely path} minimizing the rate function $J$ for a given $K$. The correction is expected to be of the same order for other models as well as for other types of options.

 For the Bachelier model, the diffusion on $[0,1]$ is $dS_u^\e=\sqrt{\e} dW_u$ and the most-likely path \eqref{eqn: MLP bachelier} when written as a path on $[0,1]$ is
 \begin{equation}
 s(u)=-\frac{3(K-s_0)}{2}u^2 + 3(K-s_0) u + s_0\ , \ u\in[0,1]\ .
 \end{equation}
 By \eqref{eqn: expect}, the price of the call option becomes
 \begin{equation}
 \label{eqn: price bachelier}
 \E\left[\left(\int_0^1 S_u^\e du -K\right)^+\right]=\int_K^\infty \PP\left(\int_0^1 S_u^\e \ du >x\right) dx\ .
 \end{equation}
Write $\Q$ for the measure with $d\Q/d\PP = \exp\left(\e^{-1/2}\int_0^1 \dot{s}(u)dW_u -\frac{\e^{-1}}{2}\int_0^1 \dot{s}(u)^2du\right)$.
Under $\Q$, the process $(W_u,u\in[0,1])$ is a Brownian motion with drift $\e^{-1/2}\dot{s}(u)$.
With this notation, \eqref{eqn: price bachelier} becomes
\begin{equation}
e^{ -\frac{\e^{-1}}{2}\int_0^1 \dot{s}(u)^2du}\int_K^\infty \E_{\Q}\left[e^{-\e^{-1/2}\int_0^1 \dot{s}(u)d\widetilde W_u}1_{\{\int_0^1 \e^{1/2}\widetilde W_u \ du + K>x\}}\right] dx
\end{equation}
where $(\widetilde W_u,u\in[0,1])$ is a standard Brownian motion under $\Q$ and we use the fact that $\int_0^1 s(u) du=K$ by definition of the most likely path.
By doing the change of variable $y=x-K$, this reduces to
\begin{equation}
 \E\left[\left(\int_0^1 S_u^\e du -K\right)^+\right]=e^{ -\frac{\e^{-1}}{2}\int_0^1 \dot{s}(u)^2du}\int_0^\infty \E_{\Q}\left[e^{-\e^{-1/2}\int_0^1 \dot{s}(u)d\widetilde W_u}1_{\{\int_0^1 \e^{1/2}\widetilde W_u \ du >x\}}\right] dx\ .
\end{equation}
The first term is $e^{-\e^{-1}J(K)}$ and gives the first order. To evaluate the second term, it is convenient to first integrate $x$ to get
\begin{equation}
\E_{\Q}\left[e^{-\e^{-1/2}\int_0^1 \dot{s}(u)d\widetilde W_u}
\left(\int_0^1 \e^{1/2}\widetilde W_u \ du \right)^+\right]\ .
\end{equation}
Note that $\dot{s}(u)=3(K-s_0)(1-u)$, therefore
\begin{equation}
\int_0^1 \dot{s}(u)d\widetilde W_u=3(K-s_0)\int_0^1 \widetilde{W}_u du\ .
\end{equation}
Write $X$ for the random variable $\int_0^1 \widetilde{W}_u du$ which is Gaussian with mean $0$ and variance $1/3$. We have
\begin{equation}
 \e^{1/2}\E_{\Q}\left[X^+\ e^{-\e^{-1/2} 3(K-s_0) X}\right]
 =\e^{3/2}\int_0^\infty y e^{-3(K-s_0) y}\ \frac{e^{-\frac{3\e y^2}{2}}}{\sqrt{2\pi/3}}dy\ .
\end{equation}
The integral is of order $1$. We conclude that
\begin{equation}
 \E\left[\left(\frac{1}{T}\int_0^T S_{u}\ du -K\right)^+\right]
 =\exp\Big(-T^{-1}J(K)+\frac{3}{2}\log T + \mathcal O(1)\Big)\ .
\end{equation}

%
%

\section{Numerical tests} \label{sec:numerics}
From Theorem \ref{thm:imp-vol-approx}, we have the following approximate formula for Asian Bachelier implied volatility:
\begin{equation}
\sigma_{b,0} = \left( \frac T{3(K - s_0)^2} \int_0^T \left[ \frac{\dot{\tilde s}(t)}{a(\tilde s(t),t)} \right]^2dt \right)^{-\frac12},
\label{eq:approxVol}
\end{equation}
where $\tilde s(t)$ is the most-likely-path for an Asian option determined by solving the variational problem
\begin{equation}
  \min_{s:\ t\mapsto s(t)} \frac12 \int_0^T \left[ \frac{\dot s(t)}{a(s(t),t)}  \right]^2 dt
\end{equation}
subject to
\begin{equation}
  \frac1T \int_0^T s(t) dt = K, \quad s(0) = s_0.
\end{equation}

We now proceed to test numerically the approximate implied volatility formula \eqref{eq:approxVol} for various definitions of the local volatility function $a(s,t)= s\,\sigma_\ell(s,t)$ in the local volatility model \eqref{eqn:local-vol}. In each case, to evaluate \eqref{eq:approxVol}, we need to compute the most-likely-path $\tilde s(t)$. To this end, we exploit the following iteration scheme.

\begin{lemma}\label{lem:recursion}
The most-likely path $\tilde s(t)$ satisfies the recursive formula
\begin{equation}
\tilde s(t) = s_0  +\frac{I(t)}{\bar I(T)}\,\left[ K-s_0\right]
\end{equation}
where
\begin{eqnarray*}
I(t)&=& \int_0^t \int_r^T\, a(\tilde s(r),r) a(\tilde s(u),u) e^{-\int_t^u \frac{a_t(\tilde s(v),v)}{a(\tilde s(v),v)}dv} \,du dr, \\
\bar I(T) &=& \frac 1T \int_0^T\,I(u)\,du.
\end{eqnarray*}
\end{lemma}

\begin{proof}
From Lemma \ref{lem:2.2}, $\tilde s(t)$ satisfies the Euler-Lagrange equation \eqref{ode} which we reiterate in the following for convenience
\beq
 \frac{d}{dt}\left(\frac{\dot s}{a} \right) - \frac{a_t}{a^2} \dot s + \frac\lambda T a = 0  \label{eq:EL2}
\eeq
with boundary conditions
\begin{equation}
   s(0) = s_0, \quad \dot s(T) = 0, 
\end{equation}
where $\lambda$ is chosen such that $1/T\int_0^T s(t) dt = K$.

For ease of notation, define $\tilde a(t):=a\left(\tilde s(t),t\right)$ and $f(s,t) = a_t/a = \p_t \log a$. Also, $\tilde f(t) = f(\tilde s(t),t)$.
Applying the integrating factor $\exp\left(-\int_0^t \tilde f(u)du\right)$ and integrating \eqref{eq:EL2} with the boundary condition $\dot s(T) = 0$ gives
\beaa
&& -e^{-\int_0^t \tilde f(v)dv} \frac{\dot s(t)}{\tilde a(t)} = -\frac{\lambda}{T}\, \int_t^T\, \tilde a(u) e^{-\int_0^u \tilde f(v)dv} \,du. 
\eeaa
It follows that
\beaa
\frac{\dot s(t)}{\tilde a(t)} = \frac{\lambda}{T}\, \int_t^T\, \tilde a(u) e^{-\int_t^u \tilde f(v)dv} \,du.
\eeaa
Rearranging and integrating again gives
\begin{equation}
 s(t) - s_0 = \frac{\lambda}{T}\,I(t), 
\end{equation}
where $I(t) = \int_0^t \int_r^T\, \tilde a(r) \tilde a(u) \exp\left(-\int_t^u \tilde f(v)dv\right) \,du dr$.
Now apply the boundary condition $1/T\int_0^T s(t) dt = K$ to get
\begin{equation}
K -s_0 =  \frac{\lambda}{T}\,\bar I(T) 
\end{equation}
and the result follows.
\end{proof}

Lemma \ref{lem:recursion} leads to an efficient fixed-point algorithm for solving for the most-likely-path.  The natural choice of first guess is  the Bachelier most-likely-path \eqref{eqn: MLP bachelier}:
\begin{equation}
 s(t) = s_0  + 3(K - s_0) \frac tT - \frac{3(K - s_0)}2 \left( \frac tT \right)^2 .
\end{equation}
The resulting algorithm typically converges sufficiently after three or four iterations.

\subsection{A time-dependent CIR model}

We consider the model
  \begin{equation}
    dS_t = e^{-\lambda t} \sigma \sqrt{S_t} \,dW_t
  \end{equation}
with $S_0=1$, $\sigma = 0.2$ and $\lambda=1$.  Thus
\beq
a(s,t) = e^{-\lambda t} \sigma \sqrt{s}.
\label{eq:CIR}
\eeq

Though we computed a quasi-closed-form for the most-likely-path for the time-homogeneous case in Example \ref{ex:cir}, in this time-inhomogeneous case, we choose to compare our approximate implied volatility formula \eqref{eq:approxVol} evaluated using the fixed-point iteration algorithm against Monte Carlo simulations generated with 1000 time steps and 2 million sample paths.  The results are shown in Figure \ref{fig:CIR}. We remark that the most-likely-path approximation in this example slightly underestimates the normal implied volatility for Asian option inferred from simulation.  

\begin{figure}[ht!]
\begin{center}
\includegraphics[width=0.7 \linewidth]{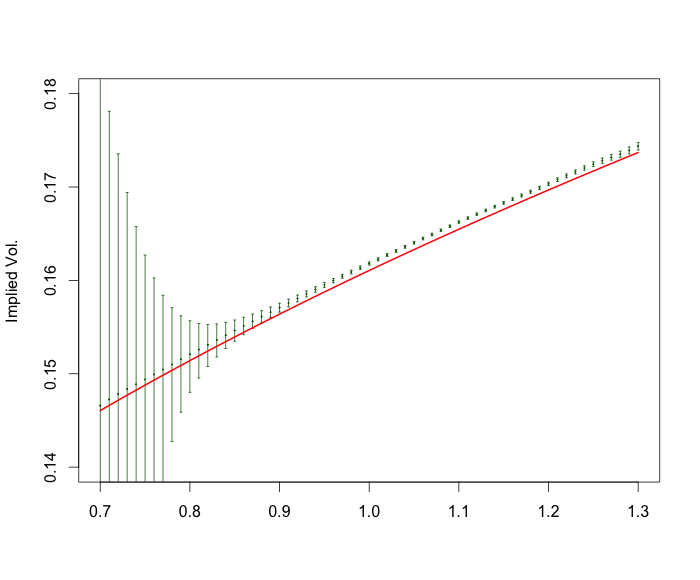}
\end{center}
\caption{The $1$-year Bachelier Asian implied volatility smile corresponding to the time-dependent CIR local volatility function \eqref{eq:CIR}.  The blue dotted line is from Monte Carlo simulation with error bars; the red solid line is the approximation $\sigma_{b,0}$. }
\label{fig:CIR}
\end{figure}

\subsection{Time-dependent quadratic local volatility}\label{sec:AQ}

Next we consider the following quadratic local volatility model:
  \beq
    dS_t = e^{-\lambda t} \sigma \left[ 1 + \psi (S_t - 1) + \frac\gamma2 (S_t - 1)^2 \right] dW_t
    \label{eq:CEVT}
  \eeq
with $\sigma=0.2$, $\psi=-0.5$, $\gamma=0.1$, and $\lambda=1$. We remark that though in this example the function $a$ grows quadratically to infinity as $|x| \to \infty$ which violates the linear growth condition \eqref{eqn: bounded} required for the theoretical argument, we did the numerical experiment for testing the applicability of the most-likely-path methodology.  

Though a closed-form solution for European options with these parameters is given in \cite{andersen}, we again resort to Monte Carlo simulation to estimate the value of Asian options in this model. Likewise, simulations are generated with 1000 time steps and 2 million sample paths. 
The results are shown in Figure \ref{fig:AQ}. Similarly, the most-likely-path approximation in this example also underestimates the normal implied volatility for Asian option inferred from simulation.

\begin{figure}[ht!]
\begin{center}
\includegraphics[width=0.7 \linewidth]{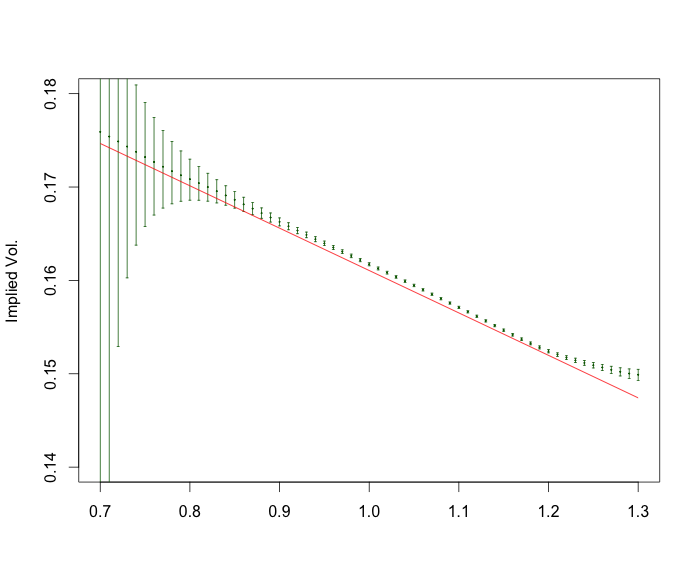}
\end{center}
\caption{The $1$-year Bachelier Asian implied volatility smile corresponding to the time-dependent quadratic local volatility function \eqref{eq:CEVT}.  The blue dotted line is from Monte Carlo simulation with error bars; the red solid line is the approximation $\sigma_{b,0}$. }
\label{fig:AQ}
\end{figure}

\subsection{Benchmark scenarios in Black-Scholes and CIR models}\label{sec:numerics-benchmark}
In this subsection, for the Black-Scholes and CIR models, we compare the most-likely-path approximation to a few existing approximations with the benchmark scenarios proposed in \cite{ge} and \cite{fmw} that were commonly used in the literatures on Asian option pricing, see for instance \cite{dassios}, \cite{dufresne}, \cite{fpp}, \cite{linetsky}, \cite{vx}.

With the approximation of Bachelier implied volatility $\sigma_{b,0}$ given in Theorem \ref{thm:imp-vol-approx}, we may calculate the approximate price of the Asian call struck at $K$ and expired at $T$ by  
\bea
e^{-rT} \left\{\sqrt{\frac v{2\pi}} e^{-\frac{(A - K)^2}{2v}} + (A - K) N\left( \frac{A - K}{\sqrt v} \right)\right\} \label{eqn:approx-Asian-call-price}
\eea
with
\beaa
&& A = \frac{S_0(e^{\mu T} - 1)}{\mu T}, \quad  
v = \frac{\sigma_{b,0}^2}{\mu^2 T^2} \left( \frac{3 - 4e^{\mu T} + e^{2\mu T}}{2\mu} + T\right) 
\eeaa
and $\mu = r - q$. Note that \eqref{eqn:approx-Asian-call-price} is indeed the price of the Asian call struck at $K$ and expired at $T$ under the Bachelier model 
\begin{equation}
dS_t = (r - q) S_t dt + \sigma dW_t
\end{equation}
with risk free rate $r$ and dividend rate $q$, both assumed constant. We remark that, since many of the benchmark scenarios are of ATM, the ATM approximate implied volatility is obtained by taking the limit of $\sigma_{b,0}$ as $K$ approaches $S_0$. Explicitly, the limits in the Black-Scholes and CIR cases are given by
\beaa
&& \lim_{K \to S_0} \sigma_{b,0} = \left\{\begin{array}{ll}
\sigma S_0 & \mbox{for the Black-Scholes model,} \\
& \\
\sigma \sqrt{S_0} & \mbox{for the CIR model.}
\end{array}\right.
\eeaa

Table \ref{tab:BS} exhibits the numerical results for the asymptotic approximation for the Asian options obtained from \eqref{eqn:approx-Asian-call-price} for the scenarios considered in \cite{fmw}. Based on the results in \cite{linetsky}, the relative discrepancies of the approximate prices are less than 1.5\% in all the seven benchmark scenarios and within 1\% for options expiring in a year. 
Table \ref{tab:CIR} shows the numerical tests for Asian option pricing in the CIR model for the seven scenarios proposed in \cite{dassios}. Data is quoted from Table 5 in \cite{fpp}. Based on the results in \cite{fpp}, the relative discrepancies of the approximate prices are less than 1\% in all the seven benchmark scenarios and within 0.6\% for options expiring in a year. We may thus arguably conclude that the short expiry approximation results of the current paper provide a reasonable approximation for Asian option prices in these models.
\begin{table}[]
\centering
\caption{Prices of Asian calls in the Black-Scholes model in the benchmark scenarios. The last four columns correspond to the approximate price from \eqref{eqn:approx-Asian-call-price} (ALW), the third order approximation from \cite{fpp} (FPP3), the precise evaluation in \cite{linetsky} (Linetsky), and the relative discrepancy of ALW to Linetsky.}
\label{tab:BS}
\begin{tabular}{|c|ccccc|cccc|}
\hline
Case & $S_0$ & $K$ & $r$ & $\sigma$ & $T$ & ALW & FPP3 & Linetsky & rel. discrp.\\ \hline \hline
1 & 2 & 2 & 0.02 & 0.1 & 1 & 0.056042 & 0.055986 & 0.055986 & 0.10\% \\ 
2 & 2 & 2 & 0.18 & 0.3 & 1 & 0.219607 & 0.218387 & 0.218387 & 0.56\% \\ 
3 & 2 & 2 & 0.0125 & 0.25 & 2 & 0.172939 & 0.172267 & 0.172269 & 0.39\% \\
4 & 1.9 & 2 & 0.05 & 0.5 & 1 & 0.195034 & 0.193164 & 0.193174 & 0.96\% \\
5 & 2 & 2 & 0.05 & 0.5 & 1 & 0.248277 & 0.246406 & 0.246416 & 0.73\% \\
6 & 2.1 & 2 & 0.05 & 0.5 & 1 & 0.308029 & 0.306210 & 0.306220 & 0.59\% \\
7 & 2 & 2 & 0.05 & 0.5 & 2 & 0.355167 & 0.350040 & 0.350095 & 1.45\% \\ \hline
\end{tabular}
\end{table}
\begin{table}
\centering
\caption{Prices of Asian calls in the CIR model in the benchmark scenarios considered in \cite{dassios}. The last four columns correspond to the approximate price from \eqref{eqn:approx-Asian-call-price} (ALW), from \cite{dassios} (DN), the third order approximation from \cite{fpp} (FPP3), and the relative discrepancy of ALW to FPP3.}
\label{tab:CIR}
\begin{tabular}{|c|ccccc|cccc|}
\hline
Case & $S_0$ & $K$ & $r$ & $\sigma$ & $T$ & ALW & DN & FPP3 & rel. discrp. \\ \hline \hline
1 & 2 & 2 & 0.02 & 0.14 & 1 & 0.055591 & 0.0197 & 0.055562 & 0.05\% \\
2 & 2 & 2 & 0.18 & 0.42 & 1 & 0.218521 & 0.2189 & 0.217874 & 0.30\% \\
3 & 2 & 2 & 0.0125 & 0.35 & 2 & 0.171331 & 0.1725 & 0.170926 & 0.24\% \\
4 & 1.9 & 2 & 0.05 & 0.69 & 1 & 0.191950 & 0.1902 & 0.190834 & 0.58\% \\
5 & 2 & 2 & 0.05 & 0.72 & 1 & 0.252333 & NA & 0.251121 & 0.48\% \\
6 & 2.1 & 2 & 0.05 & 0.72 & 1 & 0.309864 & 0.3098 & 0.308715 & 0.37\% \\
7 & 2 & 2 & 0.05 & 0.71 & 2 & 0.356411 & 0.3339 & 0.353197 & 0.91\% \\
\hline
\end{tabular}
\end{table}

%
%

\section{Conclusion}
We have derived a small time asymptotic of the price of discretely monitored Asian options up to first order. The most-likely-path approximation for continuously monitored Asian calls has been derived heuristically by taking the limit of the corresponding term from discretely monitored case  and proved rigorously using the theory of large deviation.
Numerical experiments in both time-dependent CIR and time-dependent quadratic local volatility models showed that there is room to improve the performance of the most-likely-path approximation. One possibility is to include higher order terms. Large deviation theory generically provides no insights beyond the term of exponential decay. On the other hand, it is conceivable that taking the limit of corresponding terms from the discretely monitored case may yield tractable expressions for numerical evaluations. The calculations of higher order terms are considerably more involved and henceforth were left to further work.

%
%

\section*{Acknowledgement}
We thank the anonymous referees for their careful reading and valuable comments.
We would like to express our gratitude to Jim Gatheral for numerical results and Dan Pirjol for valuable and interesting discussions, and sending us an early version of their papers. 
L.-P. A. is supported by NSF CAREER 1653602, NSF grant DMS-1513441, and a Eugene M. Lang Junior
Faculty Research Fellowship. N.-L. L. is supported by JSPS KAKENHI Grant Numbers 25780213.
T.-H. W. is partially supported by the Natural Science Foundation of China grant 11601018.

%
%
\section{Appendix - Convexity of the constrained optimization problem} \label{sec:appendix-convexity}

We analyze the convexity of the constrained optimization problem \eqref{eqn:discrete-obj}:\eqref{eqn:discrete-constr} in this appendix. Recall the Lamperti transformation $\varphi(s,t) := \int_{s_0}^s 1/a(\xi,t) d\xi$ and the objective function $D$
\begin{equation}
  D(\bs,\mt) = \frac12 \sum_{i=1}^n |\varphi(s_i,t_i) - \varphi(s_{i-1},t_{i-1})|^2.
\end{equation}
For notational simplicity, we shall write the function $a(\cdot,t_i)$ as $a_i(\cdot)$ and similarly $\varphi(\cdot,t_i)$ as $\varphi_i(\cdot)$. By straightforward calculations, the second partial derivatives of $D$ are given by
\beaa
  && \frac{\p^2 D}{\p s_i^2} = \frac2{a_i^2(s_i)} + \frac{a_i'(s_i)}{a_i^2(s_i)}[\varphi_{i-1}(s_{i-1})
  + \varphi_{i+1}(s_{i+1})-2\varphi_i(s_i)]; \\
  && \frac{\p^2 D}{\p s_n^2} = \frac1{a_n^2(s_n)} - \frac{a_n'(s_n)}{a_n^2(s_n)}[\varphi_{n-1}(s_n) - \varphi_{n-1}(s_{n-1})];  \\
  && \frac{\p^2 D}{\p s_i \p s_j} = \frac{-1}{a_i(s_i)a_j(s_j)}, \quad \mbox{ if } |i - j| = 1;   \\
  && \frac{\p^2 D}{\p s_i \p s_j} = 0, \quad \mbox{ if } |i - j| \geq 2 .
\eeaa
We decompose the Hessian matrix $H = \left[ \p_{s_i}\p_{s_j} D\right]$ as $H = H^1 + H^2$, where $H^1$ is the symmetric tridiagonal matrix with diagonal entries given by 
\beaa
&& H^1_{kk} = \frac2{a_k^2(s_k)} \quad \mbox{ for } k = 1, \cdots, n-1, \quad
H^1_{nn} = \frac1{a_n^2(s_n)}
\eeaa
and off-diagonal entries by  
\beaa
H^1_{k, k+1} = \frac{-1}{a_k(s_k) a_{k+1}(s_{k+1})} \quad \mbox{ for } k = 1, \cdots, n-1. 
\eeaa
$H^2$ is a diagonal matrix with diagonal entries given by 
\beaa
&& H^2_{kk} = \frac{a_k'(s_k)}{a_k^2(s_k)}[\varphi_{k-1}(s_{k-1}) + \varphi_{k+1}(s_{k+1}) - 2 \varphi_k(s_k)] \mbox{ for } k = 1, \cdots, n-1, \\
&& H^2_{nn} = - \frac{a_n'(s_n)}{a_n^2(s_n)}[\varphi_{n-1}(s_n) - \varphi_{n-1}(s_{n-1})].
\eeaa
We claim that $H^1$ is positively definite. Let $\hat H^1$ be the square submatrix of $H^1$ consisting of the first $n-1$ rows and columns of $H^1$.   We partition $H^1$ as
\beaa
H^1 = \left[\begin{array}{cc}
\hat H^1 & c^T \\
& \\
c & H^1_{nn}
\end{array}\right],
\eeaa
where $c = \left[ 0\, \cdots \, 0 \; H^1_{n-1, n} \right]$ is an $(n-1)$ row vector. Note that, by induction, one can show that the principle minors of $\hat H^1$ are $(k+1)\prod_{j=1}^k a_j^{-2}(s_j) > 0$,  for $k = 1, \cdots, n-1$. By applying the identity for determinant, if $A$ is invertible,  
\beaa
\det\left[\begin{array}{cc}
A & b^T \\
b & d
\end{array}\right] 
= \det(A) \times (d - b A^{-1} b^T), 
\eeaa
we calculate the determinant of $H^1$ as 
\beaa
\det(H^1) &=& \det(\hat H^1) \times \left\{H^1_{nn} - c(\hat H^1)^{-1}c^T\right\} \\
&=& \det(\hat H^1) H^1_{nn} - \left( H^1_{n-1,n} \right)^2 \det(\hat{\hat H}^1) \\
&=& \prod_{k=1}^n \frac1{a_k^2(s_k)} > 0,
\eeaa
where $\hat{\hat H}^1$ is the square submatrix of $H^1$ by deleting the last two rows and columns of $H^1$. 
Thus, by Sylvester's criterion $H^1$ is positive definite. 

Finally, since
\beaa
  && \varphi_{k-1}(s_{k-1}) + \varphi_{k+1}(s_{k+1}) - 2 \varphi_k(s_k) \nonumber \\
  && = \varphi_k'(s_k) (\Delta s_{k+1} + \Delta s_k) + o(\Delta s_{k+1}, \Delta s_k, \Delta t) \\
  && \varphi_{n-1}(s_n) - \varphi_{n-1}(s_{n-1}) = \varphi_{n-1}'(s_{n-1}) \Delta s_n + o(\Delta s_n),
\eeaa
where $\Delta s_k = s_k - s_{k-1}$, 
the entries of $H^2$ are small if the $s_i$'s are close to each other and $\Delta t$ is small. In this case, $H$ is positive definite. Hence, the objective function $D$ is convex.


\end{document}